\documentclass[11pt,letterpaper]{article}

\usepackage{tabularx}
\usepackage{microtype}
\usepackage{epsfig}
\usepackage{enumerate}
\usepackage{graphicx}
\usepackage{amsthm}
\usepackage{amssymb}
\usepackage{amsmath}
\usepackage{amsfonts}
\usepackage{mathrsfs}
\usepackage{algorithm2e}

\usepackage{graphicx}
\usepackage{tikz, subfigure}
\usetikzlibrary{decorations.shapes,decorations.pathreplacing}
\usetikzlibrary{matrix}
\usetikzlibrary{arrows,shapes,positioning}

\usepackage{vmargin}
\setmarginsrb{1in}{1in}{1in}{1in}{0pt}{0pt}{0pt}{6mm}

\usepackage{color}

\newtheorem{proposition}{Proposition}
\newtheorem{lemma}{Lemma}
\newtheorem{theorem}{Theorem}
\newtheorem{corollary}{Corollary}
\newtheorem{definition}{Definition}

\newcommand{\bw}{\mathrm{bw}}
\newcommand{\TD}{\mathcal{T}}
\newcommand{\tw}{\mathrm{tw}}
\newcommand{\pw}{\mathrm{pw}}
\newcommand{\prob}[1]{\textsc{#1}}
\newcommand{\PD}{\mathcal{P}}
\newcommand{\diam}{\text{diam}}

\newcommand{\cclass}[1]{\textnormal{\textsf{\small #1}}}
\newcommand{\cP}{\mbox{\cclass{P}}}
\newcommand{\cNP}{\mbox{\cclass{NP}}}

\newcommand{\cFPT}{\mbox{\cclass{FPT}}}

\newcommand{\cW}[1]{\mbox{\cclass{W[#1]}}}

\begin{document}

\title{Parameterized Complexity of Bandwidth on Trees}

\author{Markus Dregi\thanks{Department of Informatics, University of Bergen, Norway}
\and Daniel Lokshtanov\footnotemark[1]}

\maketitle

\begin{abstract}


The bandwidth of a $n$-vertex graph $G$ is the smallest integer $b$ such that
there exists a bijective function $f : V(G) \rightarrow \{1,...,n\}$, called a
layout of $G$, such that for every edge $uv \in E(G)$, $|f(u) - f(v)| \leq b$.
In the {\sc Bandwidth} problem we are given as input a graph $G$ and integer
$b$, and asked whether the bandwidth of $G$ is at most $b$. We present two
results concerning the parameterized complexity of the {\sc Bandwidth} problem
on trees. 

First we show that an algorithm for {\sc Bandwidth} with running time
$f(b)n^{o(b)}$ would violate the Exponential Time Hypothesis, even if the
input graphs are restricted to be trees of pathwidth at most two. Our lower
bound shows that the classical $2^{O(b)}n^{b+1}$ time algorithm by Saxe
[SIAM Journal on Algebraic and Discrete Methods, 1980] is essentially optimal.

Our second result is a polynomial time algorithm that given a tree $T$ and
integer $b$, either correctly concludes that the bandwidth of $T$ is more than
$b$ or finds a layout of $T$ of bandwidth at most $b^{O(b)}$. This is the first
parameterized approximation algorithm for the bandwidth of trees.

\end{abstract}

\section{Introduction}

A layout for a graph~$G$ is a bijective function~$\alpha: V(G) \rightarrow \{1,
\ldots, |V(G)|\}$, and the bandwidth of the layout $\alpha$ is the maximum over
all edges $uv \in E(G)$ of $|\alpha(u)- \alpha(v)|\leq b$. The bandwidth of $G$
is the smallest integer $b$ such that $G$ has a layout of bandwidth $b$. In the
{\sc Bandwidth} problem we are given as input a graph~$G$ and an integer~$b$
and the goal is to determine whether the bandwidth of $G$ is at most $b$. In
the optimization variant we are given $G$ and the task is to find a layout with
smallest possible bandwidth.

The problem arises in sparse matrix computations, where given an $n\times n$
matrix~$A$ and an integer~$k$, the goal is to decide whether there is a
permutation matrix~$P$ such that $PAP^T$ is a matrix whose all non-zero entries
lie within the $k$ diagonals on either side of the main diagonal. Standard
matrix operations such as inversion and multiplication as well as Gaussian
elimination can be sped up considerably if the input matrix~$A$ can be
transformed into a matrix~$PAP^T$ of small bandwidth \cite{George1981}. 


{\sc Bandwidth} is one of the most well-studied
$\cNP$-complete~\cite{GareyJ1979,Papadimitriou76} problems. The problem remains
$\cNP$-complete even on very restricted subclasses of trees, such as
caterpillars of hair length at most $3$~\cite{monien1986bandwidth}. Furthermore, it is
$\cNP$-hard to approximate the bandwidth within any constant factor, even on
trees~\cite{dubey2011hardness}. The best approximation algorithm for {\sc
Bandwidth} on general graphs is by Dungan and Vempala~\cite{DunaganV01}, this
algorithm has approximation ratio $(\log n)^{3}$. For trees
Gupta~\cite{Gupta00} gave a slightly better approximation algorithm with ratio
$(\log n)^{9/4}$, while for caterpillars a $O(\frac{\log n}{\log \log n})$-approximation~\cite{feige2009approximating} can be achieved.


One could argue that the {\sc Bandwidth} problem is most interesting when the
bandwidth of the graph is very small compared to the size of the graph. Indeed,
when the bandwidth of $G$ is {\em constant} the matrix operations discussed
above can be implemented in linear time. For each $b \geq 1$ it is possible to
recognize the graphs with bandwidth at most $b$ in time $2^{{\cal
O}(b)}n^{b+1}$ using the classical algorithm of Saxe~\cite{saxe1980dynamic}. At this
point it is very natural to ask how much Saxe's algorithm can be improved. Our
first main result is that assuming the Exponential Time Hypothesis of
Impagliazzo, Paturi and Zane~\cite{ImpagliazzoPZ01}, no sigificant improvement
is possible, even on very restricted subclasses of trees. In particular we show
the following theorem.
\begin{theorem}
    \label{thm:mainHardness}
    Assuming the Exponential Time Hypothesis there is no $f(b)n^{o(b)}$ time
    algorithm for {\sc Bandwidth} of trees of pathwidth at most $2$.
\end{theorem}
The proof of Theorem~\ref{thm:mainHardness} also implies that {\sc Bandwidth}
is $W[1]$-hard on trees of pathwidth at most $2$ (see~\cite{downey1999parameterized,FlumGroheBook,Niedermeier2006} for an
introduction to parameterized complexity).

As a counterweight to the bad news of Theorem~\ref{thm:mainHardness} we give
the first approximation algorithm for {\sc Bandwidth} of trees whose
approximation ratio depends only on the bandwidth $b$, and not on the size of
the graph. Specifically we give a polynomial time algorithm that given as input
a tree $T$ and integer $b$ either correctly concludes that the bandwidth of $T$
is greater than $b$ or outputs a layout of width at most $b^{O(b)}$. A key subroutine of our algorithm for trees is an approximation algorithm for the bandwidth of caterpillars with ratio $O(b^3)$. 
Our algorithm for trees outperforms the $(\log n)^{9/4}$-approximation algorithm of
Gupta~\cite{Gupta00} whenever $b = o(\frac{\log \log n}{\log \log \log n})$.
Our algorithm is the first {\em parameterized approximation} algorithm for the
{\sc Bandiwth} problem on trees, that is an algorithm with approximation ratio
$g(b)$ and running time $f(b)n^{O(1)}$. A parameterized approximation algorithm
for the closely related {\sc Topological Bandwidth} problem has been known for
awhile~\cite{Marx08}, while the existence of a parameterized approximation
algorithm for {\sc Bandwidth}, even on trees was unknown prior to this work. 

An interestng aspect of our approximation algorithm is the way we lower bound
the bandwidth of the input tree $T$. It is well known that the bandwidth of a
graph $G$ is lower bounded by its {\em pathwidth}, and by its {\em local
density}\footnote{A definition of these notions can be found in the
preliminaries}. One might wonder how far these lower bounds could be from the
true bandwidth of $G$. It was conjectured that the answer to this question is
``not too far'', in particular that any graph with pathwidth $c_1$ and local
density $c_2$ would have bandwidth at most $c_3$ where $c_3$ is a constant
depending only on $c_1$ and $c_2$. Chung and Seymour~\cite{ChungS89} gave a
counterexample to this conjecture by constructing a special kind of trees,
called {\em cantor combs}, with pathwidth $2$, local density at most $10$, and
bandwidth approximately $\frac{\log n}{\log \log n}$. Our approximation
algorithm essentially shows that the only structures driving up the bandwidth
of a tree are pathwidth, local density and cantor comb-like subgraphs.

\smallskip
\noindent
{\bf Related Work.}  
There is a vast literature on the {\sc Bandwidth} problem. For example the problem has been extensively studied from the perspective of 
approximation algoritms~\cite{dubey2011hardness,DunaganV01,Feige00,feige2009approximating,Gupta00}, 
parameterized complexity~\cite{BodlaenderFH94,GolovachHKLMS11,saxe1980dynamic}, 
polynomial time algorithms on restricted classes of graphs~\cite{assmann1981bandwidth,HeggernesKM09,kleitman1990computing,yan1997bandwidth}, 
and graph theory~\cite{ChinnOldBandwidthSurvey,ChungS89}. We focus here on the study of algorithms for {\sc Bandwidth} for small values of $b$.

Following the~$2^{O(b)}n^{b+1}$ time algorithm of Saxe~\cite{saxe1980dynamic}, published in 1980,
there was no progress on algorithms for the recognition of graphs of constant
bandwidth. With the advent of parameterized complexity in the late 80's and
early 90's~\cite{downey1999parameterized} it became an intriguing open problem whether one could
improve the algorithm of Saxe to remove the dependency on $b$ in the exponent
of $n$, and obtain a $f(b)n^{O(1)}$ time algorithm.

In a seminal paper from 1994, Bodlaender, Fellows, and
Hallet~\cite{BodlaenderFH94} proved that a number of layout problems do
not admit fixed parameter tractable algorithms unless $\cFPT=\cW{t}$ for every $t
\geq 1$, a collapse considered by many to be almost as unlikely as $\cP=\cNP$. In the
same paper Bodlaender, Fellows, and Hallet~\cite{BodlaenderFH94} claim that
their techniques can be used to show that a $f(b)n^{O(1)}$ time algorithm for
{\sc Bandwidth} would also imply $\cFPT=\cW{t}$ for every $t \geq 1$. Downey and
Fellows~(\cite{downey1999parameterized}, page 468) further claim that the techniques
of~\cite{BodlaenderFH94} imply that even fixed parameter algorithm for {\sc
Bandwidth} {\em on trees} would yield the same collapse. Unfortunately a full
version of~\cite{BodlaenderFH94} substantiating these claims is yet to appear.



\section{Preliminaries}
\label{section:prelim}

All graphs in this paper are undirected and unweighted. For a graph $G$, we
will use the notation $V(G)$ and $E(G)$ for the vertex set and edge set
respectively. Or just $V$ and $E$ whenever the graph is clear from the
context. The degree of a vertex $v$ is denoted by $\deg(v)$ and the maximum
degree in a graph by $\deg(G)$. By $\diam(G)$ we will mean the diameter of a
graph $G$.  A \emph{clique} of size $n$, denoted $K_n$ is a graph where every
pair of vertices are connected by an edge. We will use the notation $P_l$ to
describe a path of length $l$ and $\hat{P}_l$ for a specific instance of
$P_l$. When we need to index paths this will be done by superscript, i.e.
$P^i$. For two graphs $G$ and $H$, we say that $H$ is a \emph{subgraph} of $G$
if $V(H) \subseteq V(G)$ and $E(H) \subseteq E(G)$.  Furthermore, we say that
$H$ is an induced subgraph of $G$ if $V(H) \subseteq V(G)$ and $E(H) = E(G)
\cap V(H)^2$. An induced subgraph of $G$ whose vertices are $X$ is denoted by
$G[X]$. When removing a set of vertices $X$ from a graph $G$, we will use the
notation $G-X$ for the graph $G[V(G) \setminus X]$. And furthermore, if we are
removing a single vertex $v$ we will write this as $G-v$, and this is short
for $G- \{v\}$.

If a function $f$ is defined on a set $X$ and $Y \subseteq X$ we will use the
notation $f(Y)$ for $\cup_{y \in Y} f(y)$. When it is clear from the context
that we are referring to a vertex set of a graph, we will refer to just the
graph. Furthermore, when a function $f$ is defined on the vertex set of a
graph, we will sometimes use the sloppy notation $f(G)$ instead of $f(V(G))$.

For intervals of natural numbers we will use the notation $[n]$ for the
interval $[1, \dots, n]$. A $k$-coloring of a graph $G$ is a function from
$V(G)$ to $[k]$ such that two adjacent vertices are given different values.
The chromatic number of $G$, denoted $\chi(G)$ is the minimum $k$ such that
there is a $k$-coloring of $G$.

\subsection*{Graph Classes}
A tree is a connected graph without any cycles. A caterpillar is a tree $T$
with a path $B$ as a subgraph, such that all vertices of degree $3$ or more
lie on $B$. We then say that $B$ is a backbone of $T$ and every connected
component of $T-B$ is a stray or a hair. We say that a caterpillar is of stray
length $s$ if there exists a backbone such that all strays are of size at most
$s$.  An interval graph is a graph such that there exists a function from
$V(G)$ into intervals of $\mathbb{N}$ such that the images of two vertices have
a non-empty intersection if and only if the two vertices are adjacent.

\subsubsection*{Decompositions}
A \emph{tree decomposition} $\TD$ of a graph $G$ is a pair $(T, X)$ with $T = 
(I, M)$ being a tree and $X = \{X_i \mid i \in I\}$ a collection of subsets of 
$V$ such that:

\begin{enumerate}
  \item $\bigcup_{i \in I} X_i = V$,
  \item for every edge $uv$ there is a bag $X_i$ such that both $u$ and $v$ are 
    contained in $X_i$ and
  \item for every vertex $v \in V$ the set $\{i \in I \mid v \in X_i\}$ induces 
    a tree in $T$.
\end{enumerate}
The \emph{treewidth} of a tree decomposition $\TD$, denoted $\tw(G, \TD) = 
\max_{i \in I} |X_i| - 1$ and the treewidth of a graph $G$ is defined as 
$\tw(G) = \min \{ \tw(G, \TD) \mid \TD \text{ is a tree decomposition of } 
G\}$.  A \emph{path decomposition} $\PD$ of a graph is a tree decomposition 
such that $T$ is a path. And the \emph{pathwidth} of a graph $G$, denoted 
$\pw(G)$ is the minimum width over all path decompositions.

\subsubsection*{Orderings and Bandwidth}
A \emph{linear ordering} or \emph{layout} $\alpha$ of a set $S$ is a bijection
between $S$ and $[|S|]$. Given a graph $G = (V,E)$ and a linear ordering
$\alpha$ over $V$, the \emph{bandwidth} of $\alpha$ denoted $\bw(G, \alpha) =
\max_{uv \in E} |\alpha(u) - \alpha(v)|$.  And furthermore, the bandwidth of
$G$ denoted $\bw(G) = \min \{\bw(G,\alpha) \mid \alpha \text{ is a linear
ordering over } V \}$. We say that $\alpha$ is a \emph{$k$-bandwidth ordering}
of a graph $G$ if $\bw(G, \alpha) \leq k$. And we say that a bandwidth ordering
$\alpha$ of $G$ is optimal if $\bw(G, \alpha) = \bw(G)$.

Let $u$ and $v$ be a pair of vertices of a graph $G$ and $\alpha$ an ordering
of $V(G)$. We then say that $u$ is left of $v$ in $\alpha$ if $\alpha(u) <
\alpha(v)$ and that $u$ is right of $v$ if $\alpha(v) < \alpha(u)$. A sparse ordering
$\beta$ of a graph $G$ is an injective function from $V(G)$ to $\mathbb{Z}$.
And the bandwidth of a sparse ordering $\beta$ of $G$, denoted $\bw(G, \beta)
= \max_{uv \in E} |\beta(u) - \beta(v)|$. We say that a linear ordering $\alpha$ of
$G$ is a compression of a sparse ordering $\beta$ of $G$ if for every pair of
vertices $u,v$ in $G$ it holds that $\beta(u) < \beta(v)$ if and only if
$\alpha(u) < \alpha(v)$.

\begin{definition}
    \label{def:local-denisity}
    For a graph $G$ we define the local density of $G$ as
    \[D(G) = \max_{G' \subseteq G} \frac{|V(G')| - 1}{\diam(G')}.\]
\end{definition}

The following proposition will be used repeatedly in our arguments.

\begin{proposition}[Folklore] For every graph $G$ it holds that $D(G) \leq \bw(G)$ and $\pw(G) \leq \bw(G)$.
\end{proposition}

For a graph $T$, an integer $b$ and a $b$-bandwidth ordering $\alpha$ we
provide the following definitions. Given a set of vertices $Y \subseteq V(T)$
we define the \emph{inclusion interval of $Y$}, denoted $I(Y)$ as $[\min
\alpha(Y), \max \alpha(Y)]$ and for two vertices $u$ and $v$ we define
$I(u,v)$ as $I(\{u,v\})$ or equivalently $[\min\{\alpha(u),\alpha(v)\},
\max\{\alpha(u), \alpha(v)\}]$. Given a subgraph $H$ of $T$ we define $I(H)$
as $I(V(H))$. Whenever necessary, we will use subscript to avoid confusion
about which ordering is considered.

\subsubsection*{Problems}

We will differentiate the parametrized version of a problem (parameterized by
the natural parameter) from the classical one by putting a $p$ in front of the
name, i.e. \prob{$p$-Bandwidth} is the parameterized version of
\prob{Bandwidth}. We will face two other problems in this paper. The first one
is \prob{Clique}, where given a graph $G$ and an integer $k$, one is asked
whether there is a clique of size $k$ in $G$. The second one is \prob{Even
Clique}, which is an instance of \prob{Clique} where you are promised that $k$
is an even number. Both of the problems will be discussed in their parametrized
form.

\newcommand{\gcenter}{\textit{center}}
\newcommand{\gate}[1]{\Pi_{#1}}
\newcommand{\path}[1]{\hat{P}_{#1}}
\newcommand{\gin}{\textit{in}}
\newcommand{\gout}{\textit{out}}
\newcommand{\icenter}{\textit{in center}}
\newcommand{\ocenter}{\textit{out center}}
\newcommand{\gfirst}{\textit{first}}
\newcommand{\glast}{\textit{last}}
\newcommand{\lcenter}{\textit{lcenter}}
\newcommand{\rcenter}{\textit{rcenter}}

\section{Lower Bounds}
In this section we will give a reduction from \prob{$p$-Even Clique} to
\prob{$p$-Bandwidth} with a linear blowup of the parameter. For the rest of
this section we will refer to the parameter of the instance of \prob{$p$-Even
Clique} as $k$ and the parameter of the resulting \prob{$p$-Bandwidth} instance
as $b = 4k+16$. Before we continue, we introduce some definitions we will use
throughout the section. For a subpath $\path{l} = \{v_1, \dots, v_l\}$ of a graph $T$
we say that $\path{l}$ is \emph{stretched} with respect to a $b$-bandwidth
ordering $\alpha$ if $|\alpha(v_{i+1}) - \alpha(v_i)| = b$ for every $i \in [1,
    l)$.  Observe that as $\alpha$ is injective, stretched implies either
    $\alpha(v_1) < \alpha(v_2) < \dots < \alpha(v_l)$ or $\alpha(v_l) < \dots <
    \alpha(v_2) < \alpha(v_1)$.  Furthermore, we say that a path $P$
    \emph{passes through} some subgraph $H$ in $\alpha$ if $I(H) \subseteq
    I(P)$.

\subsection{A Gentle Introduction to the Reduction}
We will now give an informal description of the reduction. We hope it will
provide the reader with some intuition of why \prob{$p$-Bandwidth} is as hard
as it is. As already mentioned, the reduction will be from instances $(G,k)$ of
\prob{$p$-Even Clique} to instances $(T,b)$ of \prob{$p$-Bandwidth}. To obtain
the results of Theorem~\ref{thm:mainHardness} we must first of all ensure that
$(G,k)$ is a yes-instance if and only if $(T,b)$ is a yes-instance. And
furthermore, we require $T$ to be a tree of size polynomial in $|V(G)|$ and
$k$,
and that the path-width of $T$ is at most $2$. Last, $b$ must be of size
$O(k)$.

We start, by providing some boundaries for $b$-bandwidth orderings of $T$.
Meaning that we force specific parts of $T$ to be the leftmost and rightmost
elements of every such ordering.  This is done by introducing two stars with
$2b$ leafs and adding a path from one of the leafs of the first star to one
of the leafs of the second. The two stars will be referred to as walls and the
path between them as the main path.  Observe that for both of the walls, the
leafs must occupy the $2b$ values closest to the value of the center in any
$b$-bandwidth ordering. It follows that the main path must be within the
inclusion interval of the two walls, since otherwise the main path would be
stretched all to long at some edge passing through a wall.

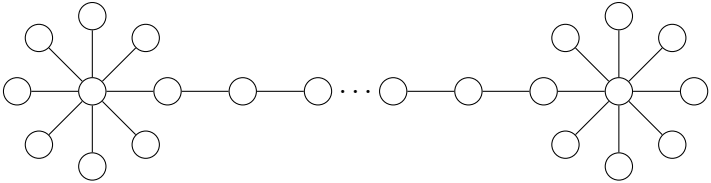
\begin{figure}[ht!]
  \centering
  \begin{tikzpicture}
    \tikzset{scale=1}
    \tikzset{Vertex/.style={shape=circle,draw,scale=1}}
    \tikzset{Edge/.style={}}
    
    \foreach \e/\x/\y [count=\k] in {
    0/0/0,
    2/2/0,
    3/3/0,
    4/4/0,
    5/5/0,
    7/7/0,
    w11/0.71/0.71,
    w12/1/0,
    w13/-0.71/0.71,
    w14/0/1,
    w15/-0.71/-0.71,
    w16/-1/0,
    w17/0.71/-0.71,
    w18/0/-1,
    w21/7.71/0.71,
    w22/8/0,
    w23/6.29/0.71,
    w24/7/1,
    w25/6.29/-0.71,
    w26/6/0,
    w27/7.71/-0.71,
    w28/7/-1
    }
    \node[Vertex] (\e) at (\x,\y) {};

    \foreach \a/\b in 
    {{0/w11},{0/w12},{0/w13},{0/w14},{0/w15},{0/w16},{0/w17},{0/w18},
    {7/w21},{7/w22},{7/w23},{7/w24},{7/w25},{7/w26},{7/w27},{7/w28},
    {w12/2},{2/3},{4/5},{5/w26}}
    \draw[Edge](\a) to node {} (\b);

    \node[] () at (3.5,0) {$\hdots$};
    
  \end{tikzpicture}
  \caption{An illustration of the walls for $b = 4$.}
  \label{figures:intro-walls}
\end{figure}

We are now controlling the first and last vertices in any $b$-bandwidth
ordering of the graph and hence it is time to start encoding our instance of
\prob{$p$-Even Clique}. To keep control, the rest of $T$ will
be attached to the main path. Before we continue, we select one of the walls
and base an ordering of the reduction graph on this selection. This wall will
from now on be referred to as the first wall and the other wall will be referred to
as the last wall. We now attach $k$ paths, from now on referred to as threads,
to the vertex of the main path that is also a leaf of the first wall. Each
thread will encode a selection of a vertex in $G$, and then we
will check whether this set of vertices in fact forms a clique or not.

\begin{figure}[ht!]
  \centering
  \begin{tikzpicture}
    \tikzset{scale=1}
    \tikzset{Vertex/.style={shape=circle,draw,scale=1}}
    \tikzset{Edge/.style={}}
 
    \foreach \e/\x/\y [count=\k] in {
    0/0/0,
    2/2/0,
    3/3/0,
    4/4/0,
    5/5/0,
    7/7/0,
    w11/0.71/0.71,
    w12/1/0,
    w13/-0.71/0.71,
    w14/0/1,
    w15/-0.71/-0.71,
    w16/-1/0,
    w17/0.71/-0.71,
    w18/0/-1,
    w21/7.71/0.71,
    w22/8/0,
    w23/6.29/0.71,
    w24/7/1,
    w25/6.29/-0.71,
    w26/6/0,
    w27/7.71/-0.71,
    w28/7/-1
    }
    \node[Vertex] (\e) at (\x,\y) {};

    \foreach \a/\b in 
    {{0/w11},{0/w12},{0/w13},{0/w14},{0/w15},{0/w16},{0/w17},{0/w18},
    {7/w21},{7/w22},{7/w23},{7/w24},{7/w25},{7/w26},{7/w27},{7/w28},
    {w12/2},{2/3},{4/5},{5/w26}}
    \draw[Edge](\a) to node {} (\b);

    \node[] () at (3.5,0) {$\hdots$};

    \foreach \e/\x/\y [count=\k] in {
    p11/1.5/1,
    p12/1.5/2,
    p1l1/1.5/3,
    p1l2/1.5/4,
    p21/2.5/1,
    p22/2.5/2,
    p2l1/2.5/3,
    p2l2/2.5/4,
    pk1/3.5/1,
    pk2/3.5/2,
    pkl1/3.5/3,
    pkl2/3.5/4
    }
    \node[Vertex] (\e) at (\x,\y) {};

    \foreach \a/\b in 
    {{w12/p11},{p11/p12},{w12/pk1},{pk1/pk2},{p1l1/p1l2},{pkl1/pkl2},
    {w12/p21},{p21/p22},{p2l1/p2l2}}
    \draw[Edge](\a) to node {} (\b);

    \node[] () at (1.5,2.6) {$\vdots$};
    \node[] () at (2.5,2.6) {$\vdots$};
    \node[] () at (3.5,2.6) {$\vdots$};

    \draw [decorate,decoration={brace,amplitude=10pt},xshift=0pt,yshift=0pt]
    (1.3,4.3) -- (3.7,4.3) node [black,midway,yshift=18pt] {$k$ paths};
  
  \end{tikzpicture}
  \caption{We will use $k$ paths to encode the selection of vertices to be in
  the clique.}
  \label{figures:reduction-frame}
\end{figure}
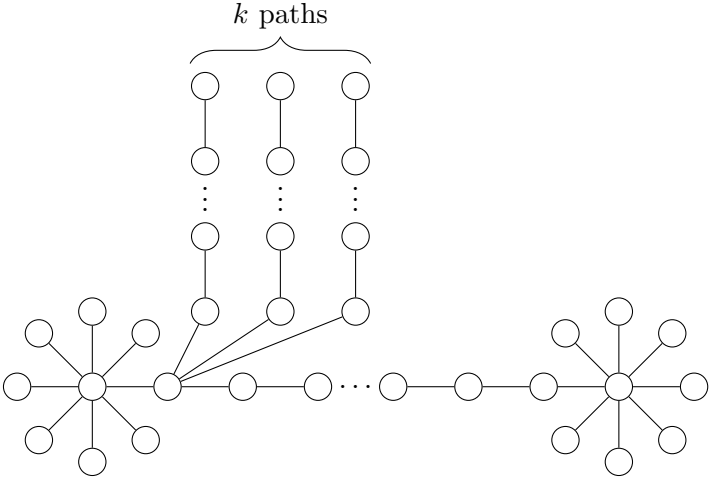

To control how information propagates through a bandwidth ordering, we
introduce gates. A $k$-gate is a vertex on the main path with $2(b-k-1)$ leafs
attached to it, that is in addition to the two neighbours it has on the main
path. The goal is to force every thread to pass through every $k$-gate. Then
every thread will position two vertices within the
positions of distance at most $b$ away from the center of the gate. And hence
there will be $2(b-k-1) + 2k + 2 = 2b$ vertices that have to be positioned
close to the center, leaving no available room.

A hole is basically two vertices on the main
path with some extra space in between. This extra space is obtained by
attaching not so many leafs to the two vertices. A knot is a large star
centered at one of the threads. The idea is that a knot requires so much space
that it cannot be positioned close to a gate. And hence, if a subpath of the
main path consists of only gates and holes, a knot that is to be positioned
within the inclusion interval of this subpath must be positioned within the
hole.

Before we start the process of embedding gadgets on the main path and the threads,
we need a guarantee ensuring that any resulting bandwidth ordering will behave
nicely. Consider the following situation, we have a graph $T$ and a
$b$-bandwidth ordering $\alpha$ of $T$. $T$ contains $k+1$ disjoint paths, one
of the paths $P$ being of length $l$ such that all the other paths are passing
through $P$ in $\alpha$. In addition there is a set of $(l-1)(b-k-1)$ vertices
$X$ disjoint from all the paths, such that the image of $X$ is contained in
the inclusions interval of $P$.
Lemma~\ref{lemma:passing-paths-are-well-behaved} then tells us that that $P$
must be stretched with respect to $\alpha$, meaning that the vertices of
$P$ appear in the same order in $T$ as in $\alpha$ up to reversion and that
the distance between two consecutive vertices is $b$. Furthermore, each of the
paths passing through will position exactly one vertex in between any two
consecutive vertices of $P$. As the reader probably can image, we will apply
this result with the main
path as $P$ and the threads as the paths passing through. This will ensure that
how and in which order the vertices appear in $\alpha$ is highly similar to how
they are ordered in $T$.

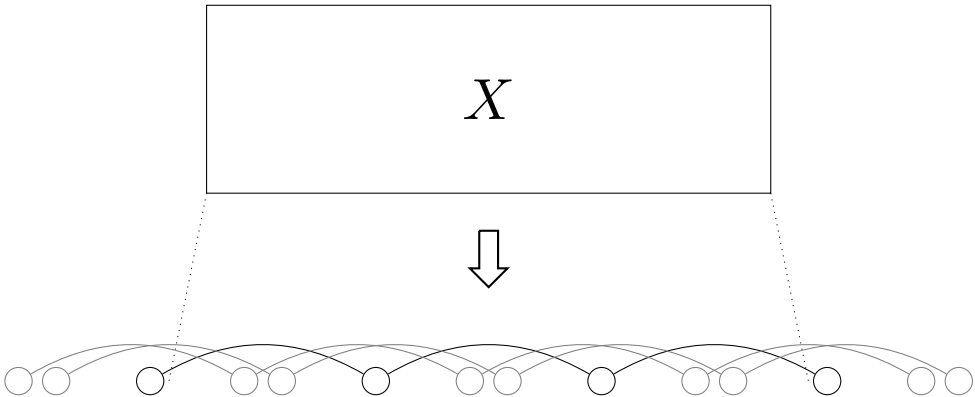
\begin{figure}[ht!]
  \centering
  \begin{tikzpicture}
    \tikzset{scale=0.25}
    \tikzset{Vertex/.style={shape=circle,draw,scale=1}}
    \tikzset{Edge/.style={}}

    \node[scale=2] at (18,15) {$X$};
    \draw (3,10) rectangle(33,20);
    \draw[dotted] (3,10) -- (1,0);
    \draw[dotted] (33,10) -- (35,0);
    \draw[thick] (17.5,8) -- (18.5,8) -- (18.5,6) -- (19,6) -- (18,5) 
    -- (17,6) -- (17.5,6) -- (17.5,8);
 
    \foreach \e/\x/\y [count=\k] in {
        m1/0/0,
        m2/12/0,
        m3/24/0,
        m4/36/0
    }
    \node[Vertex] (\e) at (\x,\y) {};

    \foreach \e/\x/\y [count=\k] in {
        p10/-7/0,
        p11/5/0,
        p12/17/0,
        p13/29/0,
        p14/41/0,
        p20/-5/0,
        p21/7/0,
        p22/19/0,
        p23/31/0,
        p24/43/0
    }
    \node[Vertex, color=gray] (\e) at (\x,\y) {};

    \draw[Edge, bend left](m1) to node {} (m2);
    \draw[Edge, bend left](m2) to node {} (m3);
    \draw[Edge, bend left](m3) to node {} (m4);

    \draw[Edge, bend left, color=gray](p10) to node {} (p11);
    \draw[Edge, bend left, color=gray](p11) to node {} (p12);
    \draw[Edge, bend left, color=gray](p12) to node {} (p13);
    \draw[Edge, bend left, color=gray](p13) to node {} (p14);

    \draw[Edge, bend left, color=gray](p20) to node {} (p21);
    \draw[Edge, bend left, color=gray](p21) to node {} (p22);
    \draw[Edge, bend left, color=gray](p22) to node {} (p23);
    \draw[Edge, bend left, color=gray](p23) to node {} (p24);
 
  \end{tikzpicture}
  \caption{An illustration of Lemma~\ref{lemma:passing-paths-are-well-behaved}.
  The black path is $P$ and the grey are the ones passing through $P$.}
  \label{figures:nicely}
\end{figure}

We will now start to embed gadgets. First we introduce three long sequences of
gates on the main path. These sequences naturally partitions our graph into nine
sectors. We will refer to them as the first wall, the
first wasteland, the first gateland, the selector, the middle gateland the
validator, the last gateland, the last wasteland and the last wall. See
Figure~\ref{figures:sectors} for an illustration. By making
the threads very long, one can force them to pass through every gate. This
together with the lemma described above implies that the sectors will appear in
the same order in any $b$-bandwidth ordering as they do in the graph up to reversion.

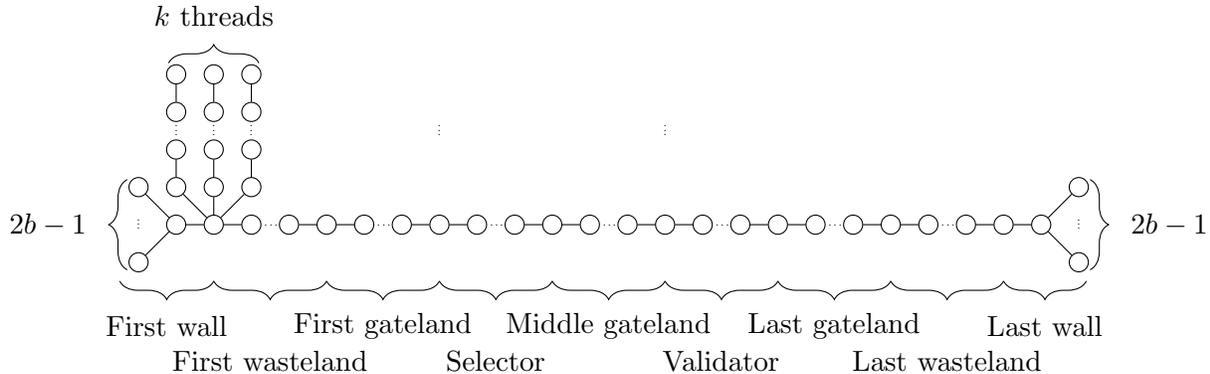
\begin{figure}[ht!]
  \centering
  \begin{tikzpicture}
    \tikzset{scale=0.5}
    \tikzset{Vertex/.style={shape=circle,draw,scale=0.7}}
    \tikzset{Edge/.style={}}
    \tikzset{Dots/.style={scale=0.4}}

    \foreach \e/\x/\y [count=\k] in {
    0/0/0,
    1/1/0,
    2/2/0,
    3/3/0,
    4/4/0,
    5/5/0,
    6/6/0,
    7/7/0,
    8/8/0,
    9/9/0,
    10/10/0,
    11/11/0,
    12/12/0,
    13/13/0,
    14/14/0,
    15/15/0,
    16/16/0,
    17/17/0,
    18/18/0,
    19/19/0,
    20/20/0,
    21/21/0,
    22/22/0,
    23/23/0,
    p11/0/1,
    p12/0/2,
    p1l1/0/3,
    p1l2/0/4,
    p21/1/1,
    p22/1/2,
    p2l1/1/3,
    p2l2/1/4,
    pk1/2/1,
    pk2/2/2,
    pkl1/2/3,
    pkl2/2/4,
    lw1/-1/1,
    lw2/-1/-1,
    rw1/24/1,
    rw2/24/-1
    }
    \node[Vertex] (\e) at (\x,\y) {};

    \foreach \a/\b in 
    {{0/1},{1/2},{3/4},{4/5},{6/7},{7/8},{9/10},{10/11},{12/13},{13/14},
    {15/16},{16/17},{18/19},{19/20},{21/22},{22/23},
    {1/p11},{p11/p12},{1/pk1},{pk1/pk2},{p1l1/p1l2},{pkl1/pkl2},
    {1/p21},{p21/p22},{p2l1/p2l2}, 
    {0/lw1},{0/lw2},{23/rw1},{23/rw2}}
    \draw[Edge](\a) to node {} (\b);

    \node[Dots] () at (0,2.6) {$\vdots$};
    \node[Dots] () at (2,2.6) {$\vdots$};
    \node[Dots] () at (1,2.6) {$\vdots$};
    \node[Dots] () at (7,2.6) {$\vdots$};
    \node[Dots] () at (13,2.6) {$\vdots$};
    \node[Dots] () at (2.55, 0) {$\dots$};
    \node[Dots] () at (5.55, 0) {$\dots$};
    \node[Dots] () at (8.55, 0) {$\dots$};
    \node[Dots] () at (11.55, 0) {$\dots$};
    \node[Dots] () at (14.55, 0) {$\dots$};
    \node[Dots] () at (17.55, 0) {$\dots$};
    \node[Dots] () at (20.55, 0) {$\dots$};
    \node[Dots] () at (-1, 0.1) {$\vdots$};
    \node[Dots] () at (24, 0.1) {$\vdots$};

    \draw [decorate,decoration={brace,amplitude=7pt},xshift=0pt,yshift=0pt]
    (-0.2,4.3) -- (2.2,4.3) node [black,midway,yshift=18pt] {$k$ threads};

    \draw [decorate,decoration={brace,amplitude=7pt},xshift=0pt,yshift=0pt]
    (-1.3,-1.2) -- (-1.3,1.2) node [black,midway,xshift=-30pt] {$2b-1$};
    
    \draw [decorate,decoration={brace,amplitude=7pt},xshift=0pt,yshift=0pt]
    (24.3,1.2) -- (24.3,-1.2) node [black,midway,xshift=30pt] {$2b-1$};
    
    \draw [decorate,decoration={brace,amplitude=7pt},xshift=0pt,yshift=0pt]
    (1,-1.5) -- (-1.5,-1.5) node [black,midway,yshift=-17pt] {First wall};
    
    \draw [decorate,decoration={brace,amplitude=7pt},xshift=0pt,yshift=0pt]
    (4,-1.5) -- (1,-1.5) node [black,midway,yshift=-30pt] {First wasteland};
    
    \draw [decorate,decoration={brace,amplitude=7pt},xshift=0pt,yshift=0pt]
    (7,-1.5) -- (4,-1.5) node [black,midway,yshift=-17pt] {First gateland};
    
    \draw [decorate,decoration={brace,amplitude=7pt},xshift=0pt,yshift=0pt]
    (10,-1.5) -- (7,-1.5) node [black,midway,yshift=-30pt] {Selector};
    
    \draw [decorate,decoration={brace,amplitude=7pt},xshift=0pt,yshift=0pt]
    (13,-1.5) -- (10,-1.5) node [black,midway,yshift=-17pt] {Middle gateland};
    
    \draw [decorate,decoration={brace,amplitude=7pt},xshift=0pt,yshift=0pt]
    (16,-1.5) -- (13,-1.5) node [black,midway,yshift=-30pt] {Validator};
    
    \draw [decorate,decoration={brace,amplitude=7pt},xshift=0pt,yshift=0pt]
    (19,-1.5) -- (16,-1.5) node [black,midway,yshift=-17pt] {Last gateland};
    
    \draw [decorate,decoration={brace,amplitude=7pt},xshift=0pt,yshift=0pt]
    (22,-1.5) -- (19,-1.5) node [black,midway,yshift=-30pt] {Last wasteland};
    
    \draw [decorate,decoration={brace,amplitude=7pt},xshift=0pt,yshift=0pt]
    (24.2,-1.5) -- (22,-1.5) node [black,midway,yshift=-17pt] {Last wall};
  \end{tikzpicture}
  \caption{The sectors of our reduction graph.}
  \label{figures:sectors}
\end{figure}

We aim at forcing a large set over vertices to 
be embedded in between the fist and the last wasteland. It follows that this part of
the main path will be stretched and every thread will position exactly one vertex
in between every two consecutive vertices of the main path. Recall that the
threads are to encode which vertices we take as our clique. This will be done
by how much of the thread is positioned within the inclusion interval of the
first wasteland before it starts its journey towards the last wasteland. And
the job of the wastelands are exactly this, to handle the slack produced by
different choices of vertices to form the clique.

We now describe how we enforce the selection of vertices in a manner that
allows us to extract this information in a useful way in the validator. First,
we order the vertices of $G$ by labeling them with numbers from $1$ to $n$.
Basically, we want there to be a linear function describing the number
of vertices positioned in the first wasteland given the label of the vertex
this thread choose. This is obtained by embedding $n$ holes within the
selector, with a certain number of gates in between every pair of consecutive
holes. Then we embed a knot on each thread. The idea is that
each thread must position its knot within a hole and every hole can contain at
most one knot. Which hole the knot is positioned within gives the vertex the
thread selects for the clique.

We should now ensure that the selected vertices forms a clique in $G$.
This is done by the validator. The validator is partitioned into $2n-1$ zones.
The first $n-1$ and last $n-1$ zones are referred to as neutral zones and nothing is
embedded on this part of the main path. The middle zone is referred to as the
validation zone. Like the selector, also the validator zone consists of $n$
holes separated by a series of gates. Now the idea is to embed the
adjacency matrix of $G$ on the threads row by row in such a way that if vertex
$i$ is selected by the thread, then the part representing row number $i$ of the
matrix is positioned within the validator zone. The matrix will be represented
as follows; Partition the subpath of the thread representing row $i$ into $n$
parts. At part number $i$ we 
embed a knot. And then, for every non-neighbour $j$ we will attach a leaf to
part $j$. What will
happen is that when the vertices are selected the corresponding holes in the
validator will be filled up by knots. And then, if two vertices are not
adjacent there will also be a leaf that should be positioned within the same
hole as a knot, and this there will not be room for. Furthermore, if a vertex is not
selected there will not be a knot in the corresponding hole so that it can contain as many
leafs as necessary. The last crucial observation is that in the neutral zone,
there is room for both leafs and knot to co-exist close in the bandwidth
ordering.

The observant reader might recall that we promised some large set of vertices
that should be embedded within the first and the last wasteland. This will be
handled by attaching paths of appropriate size right after both the first and
the second gateland. By making every hole and gate within the selector and
validator into $(k+1)$-holes and $(k+1)$-gates these paths can travel around in the
two sectors filling up the remaining space. We are now done with the informal
introduction and for the details we refer to the rest of this section.

\subsection{Tools}

In this section we give some definitions and results for bandwidth which are
crucial for our reduction. 

\begin{lemma}
    \label{lemma:no-jumping}
    Let $(T,b)$ be an instance of \prob{$p$-Bandwidth} and $\path{2}, P^1,
    \dots, P^k$ be $k+1$ disjoint subpaths of $T$. Given a $b$-bandwidth
    ordering $\alpha$ such that $P^1, \dots, P^k$ pass through $\path{2}$ and
    there is a set of vertices $X$ disjoint from $\path{2}, P^1, \dots, P^k$
    such that $|X| \geq b-k-1$ and $\alpha(X) \subseteq I(\path{2})$, then
    $|\alpha(P^i) \cap I(\path{2})| = 1$ for every $i$.
\end{lemma}

\begin{proof}
    Let $\path{2} = (u, v)$ and assume without loss of generality that $\alpha(u) <
    \alpha(v)$.  From $|I(\path{2})| \leq b+1$ and 
    \begin{align*}
        |I(\path{2})| &= |I(\path{2}) \cap \alpha(V(T))|\\
        &\geq |I(\path{2}) \cap \alpha(\bigcup P^i \cup X \cup \path{2})|\\
        &= |I(\path{2}) \cap \alpha(\bigcup P^i)| + |I(\path{2}) \cap
        \alpha(X)| + |I(\path{2})
        \cap \alpha(\path{2})|\\
        &\geq  |I(\path{2}) \cap \alpha(\bigcup P^i)| +b-k+1
    \end{align*}
    it follows that $|I(\path{2}) \cap \alpha(\bigcup P^i)| \leq k$. 
    
    Assume for a contradiction that there is a $j_1$ such that
    $|\alpha(P^{j_1}) \cap I(\path{2})| \neq 1$. Then, since $|I(\path{2}) \cap
    \alpha(\bigcup P^i)| \leq k$ it follows that there is a $j_2$ such that
    $|\alpha(P^{j_2}) \cap I(\path{2})| = 0$.  For a path $P^i$ let $(v^i_l, v^i_r)$
    maximize $\alpha(v^i_l)$ among the edges in $P^i$ with $\alpha(v^i_l) <
    \alpha(u)$ and $\alpha(v) < \alpha(v^i_r)$.  Let $P^j$ be the path
    minimizing $\alpha(v^j_l)$ among all paths $P^i$ such that $|\alpha(P^i)
    \cap I(\path{2})| = 0$. It follows that for every path $P^i$ either
    $|\alpha(P^i) \cap I(\path{2})| \geq 1$ or $|\alpha(P^i) \cap I(v^j_l, u)| \geq
    1$. Hence for each $i$ it holds that $|I(v^j_l, v^j_r) \cap \alpha(P^i)|
    \geq 1$. Furthermore, observe that $|I(v^j_l, v^j_r) \cap \alpha(P^j)| \geq
    2$. It follows that
    \begin{align*}
        |I(v^j_l, v^j_r)| &\geq |I(v^j_l, v^j_r) \cap \alpha(X)| + |I(v^j_l, v^j_r) \cap
        \alpha(\path{2})| + |I(v^j_l, v^j_r) \cap \alpha(\bigcup P^i)|\\
        &\geq (b-k-1) + 2 + (k+1) \\
        &\geq b+2 
    \end{align*} 
    Observe that $X$, $\path{2}$ and $\bigcup P^i$ are disjoint and hence the
    first line above is valid. Since $(v^j_l,v^j_r)$ is an
    edge in $T$ and $|I(v^j_l,v^j_r)| \geq b+2$ we have a contradiction to 
    $\alpha$ being a $b$-bandwidth ordering and hence our proof is complete.
\end{proof}

\begin{corollary}
    \label{corollary:passing-P2-requires-space}
    Let $(T,b)$ be an instance of \prob{$p$-Bandwidth} and $\path{2}, P^1, \dots,
    P^k$ be $k+1$ disjoint subpaths of $T$. Given a $b$-bandwidth ordering
    $\alpha$ such that $P^1, \dots, P^k$ pass through $\path{2}$ and there is a set
    of vertices $X$ disjoint from $\path{2}, P^1, \dots, P^k$ such that $|X| \geq
    b-k-1$ and $\alpha(X) \subseteq I(\path{2})$, then $|X| = b-k-1$.
\end{corollary}

\begin{proof}
    Assume for a contradiction that $|X| \geq b-k$. Apply
    Lemma~\ref{lemma:no-jumping} to obtain $|\alpha(P^i) \cap I(\path{2})| = 1$
    for every $i$. It follows that
    \begin{align*}
        |I(\path{2})|  &\geq |I(\path{2}) \cap \alpha(X \cup \path{2} \cup \bigcup
        P^i)|\\
        &\geq |I(\path{2}) \cap \alpha(X)| +
        |I(\path{2}) \cap \alpha(\path{2})| +
        |I(\path{2}) \cap \alpha(\bigcup P^i)| \\
        &\geq (b-k) + 2 + k\\
        &\geq b+2.
    \end{align*}
    which is a contradiction to $\alpha$ being a $b$-bandwidth ordering.
\end{proof}

\begin{lemma}
    \label{lemma:passing-paths-are-well-behaved}
    Let $(T,b)$ be an instance of \prob{$p$-Bandwidth} and $\path{l}, P^1,
    \dots, P^k$ be $k+1$ disjoint subpaths of $T$. Given a $b$-bandwidth
    ordering $\alpha$ such that $P^1, \dots, P^k$ pass through $\path{l}$ and
    there is a set of vertices $X$ disjoint from $\path{l}, P^1, \dots, P^k$
    such that $|X| \geq (l-1)(b-k-1)$ and $\alpha(X) \subseteq I(\path{l})$,
    then $\path{l}$ is stretched with respect to $\alpha$ and $|P^i \cap
    I(\path{2})| = 1$ for every $i$ and every $\path{2} \subseteq \path{l}$.
\end{lemma}

\begin{proof}
    We start by proving $\alpha(v_1) < \alpha(v_2) < \dots < \alpha(v_l)$ or
    $\alpha(v_l) < \dots < \alpha(v_2) < \alpha(v_1)$. Assume otherwise for a
    contradiction. Then there exists three vertices $v_{j-1}, v_j$ and
    $v_{j+1}$ such that either $\max \{\alpha(v_{j-1}), \alpha(v_{j+1})\} <
    \alpha(v_j)$ or $\alpha(v_j) < \min \{\alpha(v_{j-1}), \alpha(v_{j+1})\}$. Since all
    properties of the lemma is preserved with respect to reversing $\alpha$ we
    can assume without loss of generality that $\min \{\alpha(v_{j-1}),
    \alpha(v_{j+1})\} < \alpha(v_j)$. We define a function $f : 2^{\path{l}}
    \setminus \{\path{l}\} \rightarrow \path{l}$ as $f(B) = v_j$ such that $j =
    \min\left\{ i \mid v_i \in \path{l} \setminus B \text{ and } \{v_{i-1},
    v_{i+1}\} \cap B \neq \emptyset \right\}$. In other words, $f$ gives you
    the smallest indexed vertex in the open neighbourhood of $B$. Notice that
    since $\path{l}$ is connected $f$ is a well-defined function.  We will now
    define $a_1, \dots, a_t$ and $B_1, \dots, B_t$.  First let $a_1 =
    \alpha^{-1}(\min \left\{ \alpha(\path{l}) \right\})$ and $B_1 = \left\{ a_1
    \right\}$.  Then we let $a_i = f(B_{i-1})$ and $B_i = I(a_1, a_i) \cap
    \path{l}$ as long as $B_{i-1} \neq \path{l}$. Observe that $B_{i-1} \subset
    B_i$.
        
    First we will prove that $t < l$.  Assume otherwise for a contradiction,
    clearly then $t = l$. It follows by the construction and our assumption
    that $\left\{ a_1, \dots, a_i \right\} = B_i$ for every $i$. And by a
    simple induction we get that $T[\left\{ a_1, \dots, a_i \right\}]$ is
    connected, since this clearly holds for $i = 1$ and for $i > 1$ observe
    that $a_i$ has a neighbour in $B_{i-1}$ by construction. Let $c$ be so
    that $a_c = v_j$. Since $v_j$ is separating $v_{j-1}$ and $v_{j+1}$ in
    $\path{l}$ and $v_j \notin B_{c-1}$ it follows that $\left\{ v_{j-1},
    v_{j+1} \right\} \not\subseteq B_{c-1}$.  Furthermore, since
    $\max \{\alpha(v_{j-1}), \alpha(v_{j+1})\} < \alpha(v_j)$ it holds that $\left\{
    v_{j-1}, v_j, v_{j+1} \right\} \subseteq B_c$.  But this contradicts
    $\left\{ a_1, \dots, a_i \right\} = B_i$ and hence we know that $t < l$.
    It follows, due to the pidgin hole principle, that there is a $d$ such that $|I(a_{d-1}, a_d) \cap \alpha(X)|
    > b-k-1$.  By construction there is a neighbour $a'$ of $a_d$ among $a_1,
    \dots, a_{d-1}$. Observe that $|I(a', a_d) \cap \alpha(X)| > b-k-1$ and
    apply Corollary~\ref{corollary:passing-P2-requires-space} with $\path{2} =
    (a', a_d)$ to obtain a contradiction. Hence we can conclude that
    $\alpha(v_1) < \alpha(v_2) < \dots < \alpha(v_l)$ or $\alpha(v_l) < \dots
    < \alpha(v_2) < \alpha(v_1)$.

    We will now prove $|P^i \cap I(\path{2})| = 1$ for every $i$ and every
    $\path{2} \subseteq \path{l}$. Observe that if there is a $\path{2}$ such
    that $|I(\path{2}) \cap \alpha(X)| \neq b-k-1$, then there is a $\path{2}'$
    such that $I(\path{2}') \cap \alpha(X)| > b-k-1$. But this contradicts
    Corollary~\ref{corollary:passing-P2-requires-space} and hence we get that
    $|I(\path{2}) \cap \alpha(X)| = b-k-1$ for every $\path{2} \subseteq
    \path{l}$ and then it follows directly from Lemma~\ref{lemma:no-jumping}
    that $|P^i \cap I(\path{2})| = 1$ for every $\path{2} \subseteq \path{l}$.
    Hence \begin{align*} |I(\path{2})| &\geq |I(\path{2}) \cap \alpha(X \cup
    \path{2} \cup \bigcup P^i)| \\ &\geq |I(\path{2}) \cap \alpha(X)| +
    |I(\path{2}) \cap \alpha(\path{2})| + |I(\path{2}) \cap \alpha(\bigcup
    P^i)|\\ &\geq b-k-1 + 2 + k\\ &\geq b+1 \end{align*} and it follows that
    $\path{l}$ is stretched with respect to $\alpha$.
\end{proof}

\begin{corollary}
    \label{corollary:passing-paths-requires-space}
    Let $(T,b)$ be an instance of \prob{$p$-Bandwidth} and $\path{l}, P^1, \dots,
    P^k$ be $k+1$ disjoint subpaths of $T$. Given a $k$-bandwidth ordering
    $\alpha$ such that $P^1, \dots, P^k$ passes through $\path{l}$ and there is a
    set of vertices $X$ disjoint from $\path{l}, P^1, \dots, P^k$ such that $|X|
    \geq (l-1)(b-k-1)$ and $\alpha(X) \subseteq I(\path{l})$, then $|X| =
    (l-1)(b-k-1)$. 
\end{corollary}

\begin{proof}
    Assume for a contradiction that $|X| > (l-1)(b-k-1)$. Then there is a
    $\path{2} \subseteq \path{l}$ such that $|X \cap I(\path{2})| \geq b-k$ which
    is a contradiction by Corollary~\ref{corollary:passing-P2-requires-space}.
\end{proof}

\subsection{Gadgets}
We will now introduce the gadgets used for the reduction. They will all be
defined on paths of various lengths. And later on when we say that a gadget is
embedded on some path, this means that the path referred to together with some
of its neighbours is an instantiation of the gadget.

\begin{definition}
    \label{definition:functioning-gadget}
    Let $(T,b)$ be an instance of $p$-Bandwidth and $H$ be a subgraph of $T$
    with a vertex labeled $\gin$ and another vertex labeled $\gout$. We say
    that $H$ is \emph{functioning} in $T$ if $T$ contains two walls $W_{\gin}$
    and $W_{\gout}$ such that 
    \begin{itemize}
        \item $W_{\gin}, W_{\gout}$ and $H$ are disjoint,
        \item there is a path $P_{\gin}$ from $\gin$ to $W_{\gin}$ avoiding $(H -
            \gin)$ and $W_{\gout}$ and
        \item there is a path $P_{\gout}$ from $\gout$ to $W_{\gout}$ avoiding $(H -
            \gout)$, $W_{\gin}$ and $P_{\gin}$.
    \end{itemize}
    If $H$ is functioning in $T$ let $W_{\gin}(H, T), W_{\gout}(H, T),
    P_{\gin}(H, T)$ and $P_{\gout}(H, T)$ denote a witness of this.
\end{definition}

\subsubsection*{Walls}
A \emph{wall} is a star with $2b$ leaves. The high degree vertex of a wall $W$
will be referred to as the \emph{center} of the wall. We will turn the
endpoints of the main path into walls to control the endpoints of all valid
$b$-bandwidth orderings. The next lemma gives us this behaviour.

\begin{lemma}
    \label{lemma:walls}
    Let $(T, b)$ be an instance of \prob{$p$-Bandwidth} such that $T$ contains
    two disjoint walls $W_1$ and $W_2$ with centers $c_1$ and $c_2$ as
    subgraphs.  Let $H$ be a connected component of $T-(W_1 \cup W_2)$
    connected by edges to both walls in $T$. Then, for any $b$-bandwidth
    ordering $\alpha$ of $T$ and any vertex $v \in H$ it
    follows that $\alpha(v) \in I(c_1, c_2)$.
\end{lemma}

\begin{proof}
    Assume without loss of generality that $\alpha(c_1) < \alpha(c_2)$. For a
    contradiction, assume that $\alpha(v) < \alpha(c_1)$. Let $u_l$ be the leaf
    in $W_1$ minimizing $\alpha$ and $u_r$ the leaf maximizing $\alpha$.
    Furthermore, let $P^1$ be a path from $v$ to $c_2$ in $T[V(H) \cup
    W_2]$ and $\path{3}$ the path $(u_l, c_1, u_r)$.  Observe that $P^1$ passes
    through $\path{3}$, since $\alpha(W_1) = [\alpha(c_1)-b, \alpha(c_1) + b]$. Let $X = V(W_1) - \path{3}$ and note that $|X| = 2b-2$. Apply
    Corollary~\ref{corollary:passing-paths-requires-space} on $\path{3}, P^1$ and
    $X$ to obtain a contradiction, since $(3-1)(b-1-1) = 2b-4 < 2b-2 =
    |X|$.
    For $\alpha(v) > \alpha(c_2)$ we apply a symmetric argument and hence our
    proof is complete.
\end{proof}

\subsubsection*{Gates}
For an integer $k \geq 0$ a \emph{$k$-gate}, denoted $\gate{k}$, is a star with
$2(b-k)$ leaves. The function of the $k$-gate will be to reduce the number of
paths passing this point to at most $k$. The high degree vertex of the star
will be referred to as the \emph{center} of the gate.  In
addition one leaf will be labeled $\gin$ and another labeled $\gout$. 

\begin{figure}[ht!]
  \centering
  \begin{tikzpicture}[baseline=(0.base)]
    \tikzset{scale=2}
    \tikzset{Vertex/.style={shape=circle,draw,scale=0.7,minimum size=20pt}}
    \tikzset{Edge/.style={}}

    \node[Vertex, label=below:$\gcenter$] (center) at (0,0) {};
    \node[Vertex, label=below:$\textit{in}$] (left) at (-1,0) {};
    \node[Vertex, label=below:$\textit{out}$] (right) at (1,0) {};
    \node[] () at (0,0.9) {$\dots$};
    
    \foreach \e/\x/\y [count=\k] in {
    1/0.86/0.5,
    2/0.5/0.86,
    3/-0.5/0.86,
    4/-0.86/0.5
    }
    \node[Vertex] (\e) at (\x,\y) {};

    \foreach \a/\b in 
    {{center/left},{center/right},{center/1},{center/2},{center/3},{center/4}}
    \draw[Edge](\a) to node {} (\b);
  \end{tikzpicture}
  \caption{A $k$-gate with the special vertices marked with tags below.}
  \label{figures:kgate}
\end{figure}
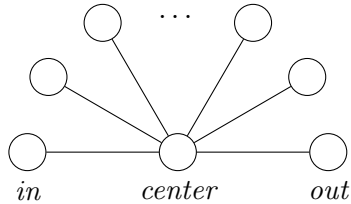

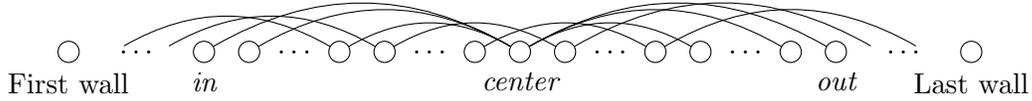
\begin{figure}[ht!]
  \centering
  \begin{tikzpicture}[baseline=(0.base)]
    \tikzset{scale=0.6}
    \tikzset{Vertex/.style={shape=circle,draw,scale=0.4,minimum size=20pt}}
    \tikzset{Edge/.style={}}

    \node[Vertex, label=below:$\gcenter$] (center) at (0,0) {};
    \node[Vertex, label=below:$\textit{in}$] (left) at (-7,0) {};
    \node[Vertex, label=below:$\textit{out}$] (right) at (7,0) {};
    \node[Vertex, label=below:$\text{First wall}$] (fw) at (-10,0) {};
    \node[Vertex, label=below:$\text{Last wall}$] (lw) at (10,0) {};

    \node[] () at (-8.5,0) {$\hdots$};
    \node[] () at (8.5,0) {$\hdots$};
    \node[] () at (-2,0) {$\hdots$};
    \node[] () at (-5,0) {$\hdots$};
    \node[] () at (5,0) {$\hdots$};
    \node[] () at (2,0) {$\hdots$};
    
    \foreach \e/\x/\y [count=\k] in {
    l1/-1/0,
    l2/-3/0,
    l3/-4/0,
    l4/-6/0,
    r1/1/0,
    r2/3/0,
    r3/4/0,
    r4/6/0
    }
    \node[Vertex] (\e) at (\x,\y) {};

    \node[] (l8) at (-8,0) {};
    \node[] (l9) at (-9,0) {};
    \node[] (r8) at (8,0) {};
    \node[] (r9) at (9,0) {};

    \draw[Edge, bend left](left) to node {} (center);
    \draw[Edge, bend left](l4) to node {} (center);
    \draw[Edge, bend left](l3) to node {} (center);
    \draw[Edge, bend right](right) to node {} (center);
    \draw[Edge, bend right](r4) to node {} (center);
    \draw[Edge, bend right](r3) to node {} (center);
    \draw[Edge, bend left](l2) to node {} (r1);
    \draw[Edge, bend left](l1) to node {} (r2);
    \draw[Edge, bend left](l8) to node {} (l2);
    \draw[Edge, bend left](l9) to node {} (l3);
    \draw[Edge, bend left](r1) to node {} (r8);
    \draw[Edge, bend left](r2) to node {} (r9);

  \end{tikzpicture}
  \caption{An illustration of Lemma~\ref{lemma:gates}.}
  \label{figures:kgate-function}
\end{figure}

\begin{lemma}
    \label{lemma:gates}
    Let $(T,b)$ be an instance of \prob{$p$-Bandwidth} such that $T$ contains a
    gate $\gate{k}$ and paths $P^1, \dots, P^k$ as disjoint subgraphs with
    $\gate{k}$ being functioning in $T - (\bigcup P^i)$. Given a $b$-bandwidth
    ordering $\alpha$ such that $\max\{\alpha(W_{\gin}(\gate{k}, T - \bigcup
    p^i))\} < \min\{\alpha(W_{\gout}(\gate{k}, T - \bigcup p^i))\}$ and every path
    $P^i$ passes through the gate it follows that:
    \begin{enumerate}[(I)]
        \item $\alpha(N[\gcenter]) \subseteq B \subseteq \alpha(\bigcup P^i \cup
            N[\gcenter])$,
        \item $\alpha(\gin) < \alpha(\gcenter) < \alpha(\gout)$ and
        \item $|\alpha(P^i) \cap B_l| = |\alpha(P^i) \cap B_r| = 1$ for every $i \in [1, k]$
    \end{enumerate}
    for $c = \alpha(\gcenter)$, $B = \left[c-b, c+b\right]$, $B_l = \{i \in B
    \mid i < c\}$ and $B_r = \{i \in B \mid c < i\}$.
\end{lemma}

\begin{proof}
    We start by proving \textit{(III)}. For every path $P^i$ we know that there
    are $u,v \in P^i$ such that $\alpha(u) < \min \alpha(\gate{k})$ and
    $\max \alpha(\gate{k}) < \alpha(v)$. Assume that $u\notin B_l$ and follow
    the path from $u$ to $v$ until you reach the first vertex $u'$ such that
    $\alpha(u') \geq c-b$. Let $u''$ be the vertex we reached right before
    $u'$. From the definition of $\alpha$ it follows that $\alpha(u') -
    \alpha(u'') \leq b$ and hence $u' \in B_l$ and $|P^i \cap B_l| = 1$.
    Reverse $\alpha$ and apply the argument on the path from $v$ to $u$ to
    obtain $|P^i \cap B_r| = 1$. 

    We continue by proving \textit{(I)}. It follows directly from the fact that
    $\bw(T, \alpha) \leq b$ that $N[\gcenter] \subseteq B$. Since $|B \cap
    \left(\bigcup P^i \cup N\left[ \gcenter \right] \right)| = |B \cap \bigcup
    P^i| + |B \cap N\left[ \gcenter \right]| = 2k + 2(b-k)+1 = 2b+1$ and $|B| =
    2b+1$ it follows that $B \subseteq \bigcup P^i \cup N\left[ \gcenter
    \right]$. It remains to prove \textit{(II)}. Observe that $\max
    \alpha(W_{\gin}) < \min \{\alpha(\gin), \alpha(\gcenter)\}$ by
    Lemma~\ref{lemma:walls}. Assume for a contradiction
    that $\alpha(\gin) > \alpha(\gcenter)$.  Since $P_{\gin}(\gate{k},
    T-(\bigcup P^i))$ is a path from $\gin$ to $W_{\gin}(\gate{k}, T - (\bigcup
    P^i))$ and the bandwidth of $\alpha$ is $b$ it follows that $|B
    \cap W_{\gin}(\gate{k}, T-(\bigcup P^i)| \geq 2$, but this contradicts
    $\textit{(I)}$ and hence $\alpha(\gin) < \alpha(\gcenter)$. A symmetric
    argument gives us $\alpha(\gcenter) < \alpha(\gout)$ and our proof is
    complete.
\end{proof}

\subsubsection*{Knots and Holes}
Assuming $b \geq 2k+14$ and $b$ to be dividable by 4 we give the following two
definitions.  A \emph{$k$-knot} is a path $P = (\gfirst, \gcenter, \glast)$ with
$\frac{3}{2}b-k-1$ leaves attached to $\gcenter$. A \emph{$k$-hole} consists of
a path $P = (\gin, \icenter, \ocenter, \gout)$ with
$\frac{3}{4}b-k-1$ leaves attached to both $\icenter$ and $\ocenter$.

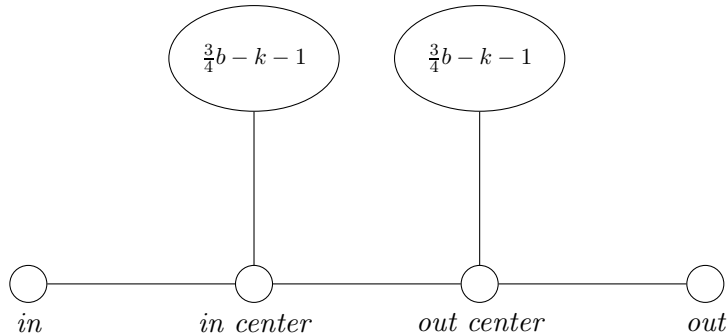
\begin{figure}[ht!]
    \centering
    \begin{tikzpicture}
        \tikzset{scale=3}
        \tikzset{Vertex/.style={shape=circle,draw,scale=0.7,minimum size=20pt}}
        \tikzset{LBag/.style={shape=ellipse,draw,scale=0.8,minimum
        size=50pt,minimum width=70pt}}
        \tikzset{Edge/.style={}}

        \node[Vertex, label=below:$\gin$] (left) at (0,0) {};
        \node[Vertex, label=below:$\icenter$] (lcenter) at (1,0) {};
        \node[Vertex, label=below:$\ocenter$] (rcenter) at (2,0) {};
        \node[Vertex, label=below:$\gout$] (right) at (3,0) {};

        \node[LBag] (lcenterl) at (1,1) {$\frac{3}{4}b-k-1$};
        \node[LBag] (rcenterl) at (2,1) {$\frac{3}{4}b-k-1$};

        \foreach \a/\b in {{left/lcenter},{lcenter/rcenter},{rcenter/right},
        {lcenter/lcenterl},{rcenter/rcenterl}} 
        \draw[Edge](\a) to node {} (\b);
    \end{tikzpicture}
    \caption{A hole. The ellipse shaped vertices represent some number of leafs.}
    \label{figures:hole}
\end{figure}

\begin{lemma}
    \label{lemma:r-hole}
    Let $(T,b)$ be an instance of \prob{$p$-Bandwidth} such that $T$ contains
    a $k$-hole $H$ and paths $P^1, \dots, P^k$ with a
    $k$-knot $K$ embedded on one of the paths as disjoint subgraphs with $H$
    being functioning in $T - (\bigcup P^i)$.  Given a $b$-bandwidth ordering
    $\alpha$ such that $P^1, \dots, P^k$ passes though $H$,
    $\max\{\alpha(W_{\gin}(\gate{k}, T - \bigcup p^i))\} <
    \min\{\alpha(W_{\gout}(\gate{k}, T - \bigcup p^i))\}$ and $I(K \cup H) \subset
    I(\gin, \gout)$ it holds that

    \begin{enumerate}[(I)]
        \item $\alpha(\gin) < \alpha(\icenter) < \alpha(\ocenter) <
            \alpha(\gout)$,
        \item $|I(\path{2}) \cap \alpha(P^i)| = 1$ 
            for every $i$ and every $\path{2} \subset \left(
            \gin, \icenter, \ocenter, \gout \right)$ and
        \item $\alpha(\icenter) < \alpha(\gcenter) < \alpha(\ocenter)$.
    \end{enumerate}
\end{lemma}

\begin{proof}
    First we prove the correctness of \textit{(I)} and \textit{(II)}. Let
    $\path{4}
    = (\gin, \icenter, \ocenter, \gout)$ and let $X_c, X_i$ and $X_o$
    be the set of leaves attached to $\gcenter, \icenter$ and $\ocenter$
    respectively. Apply Lemma~\ref{lemma:passing-paths-are-well-behaved} with
    $X = X_i \cup X_c \cup X_o$  to obtain \textit{(II)} and either
    $\alpha(\gin) < \alpha(\icenter) < \alpha(\ocenter) <
    \alpha(\gout)$ or $\alpha(\gout) < \alpha(\ocenter) <
    \alpha(\icenter) < \alpha(\gin)$ since $|X| = 2\left(\frac{3}{4}b-k-1
    \right) + \frac{3}{2}b -k-1 = (4-1)(b-k-1)$. Assume for a contradiction
    that $\alpha(\gout) < \alpha(\ocenter) <
    \alpha(\icenter) < \alpha(\gin)$. Then there is a vertex $v \in P_{\gin}(H,
    T - (\bigcup P^i))\cap \alpha^{
    -1}(I(H)) \setminus \left\{ \gin \right\}$. Apply
    Corollary~\ref{corollary:passing-paths-requires-space} with $X = X_i \cup
    X_c \cup X_o \cup \left\{ v \right\}$ to get a contradiction and hence
    \textit{(I)} holds. 

    It remains to prove \textit{(III)}. Assume for a contradiction that
    $\alpha(\gcenter) \notin I(\icenter, \ocenter)$.  Furthermore, assume
    without loss of generality that $\alpha(\gcenter) \in I(\gin,
    \icenter)$. It follows from
    Lemma~\ref{lemma:passing-paths-are-well-behaved} that $\path{4}$ is stretched
    and hence $X_l \cup X_m \subseteq I(\path{3})$ for $\path{3} = (\gin, \icenter,
    \ocenter)$. Apply Corollary~\ref{corollary:passing-paths-requires-space}
    with $X = X_i \cup X_c$ to obtain a contradiction since $|X| =
    \frac{3}{4}b -k - 1 + \frac{3}{2}b - k - 1 = \frac{9}{4}b - 2k - 2 >
    (3-1)(b-k-1)$ and hence our proof is complete. 
\end{proof}

\subsection{The Reduction}

We will now give a reduction from an instance $(G, k)$ of \prob{$p$-Even
Clique} to an instance $(T, b)$ of \prob{$p$-Bandwidth}. The correctness and
implications will be given in the two following sections. The resulting
instance $T$ can be divided into eleven parts. Nine of them lie on the main
path and will in the future be referred to as the sectors of the main path.
The nine sectors are the first wall, the first wasteland, the first gateland, the
selector, the middle gateland, the validator, the last gateland, the last
wasteland and the last wall. The two other components will be referred to as
threads and fillers. Each of the components have a specific purpose with
respect to how a $b$-bandwidth ordering can be. The walls will force
everything else to be positioned within them. The threads are $k$ paths
attached to the first wasteland and each of them represents a vertex in the
supposed clique in $G$. To encode how $G$ looks like we attach leaves to the
threads, which will be referred to as the dangelments of the threads. How much
of a thread that is in the inclusion interval of the first wasteland decides
which vertex in $G$ this thread represents. To propagate this information the
threads are made so long that they will have to enter the inclusion interval
of the last wasteland. The selectors job is to make sure that the decisions
made by the threads are unique and valid. The validator will verify that the
selected vertices in fact is a clique. And the fillers and the gatelands will
control how information propagates between the other components.  When
describing the components on the main path we will assume the vertices of the
path to be named $u_1, \dots$, with $u_1$ being the leftmost vertex in
Figure~\ref{figures:reduction-frame}.

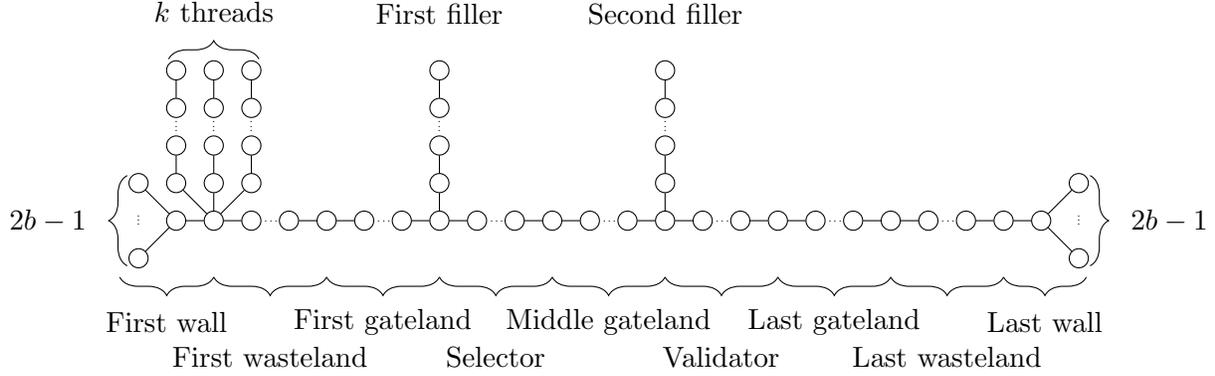
\begin{figure}[ht!]
  \centering
  \begin{tikzpicture}
    \tikzset{scale=0.5}
    \tikzset{Vertex/.style={shape=circle,draw,scale=0.7}}
    \tikzset{Edge/.style={}}
    \tikzset{Dots/.style={scale=0.4}}

    \foreach \e/\x/\y [count=\k] in {
    0/0/0,
    1/1/0,
    2/2/0,
    3/3/0,
    4/4/0,
    5/5/0,
    6/6/0,
    7/7/0,
    8/8/0,
    9/9/0,
    10/10/0,
    11/11/0,
    12/12/0,
    13/13/0,
    14/14/0,
    15/15/0,
    16/16/0,
    17/17/0,
    18/18/0,
    19/19/0,
    20/20/0,
    21/21/0,
    22/22/0,
    23/23/0,
    p11/0/1,
    p12/0/2,
    p1l1/0/3,
    p1l2/0/4,
    p21/1/1,
    p22/1/2,
    p2l1/1/3,
    p2l2/1/4,
    pk1/2/1,
    pk2/2/2,
    pkl1/2/3,
    pkl2/2/4,
    f11/7/1,
    f12/7/2,
    f13/7/3,
    f14/7/4,
    f21/13/1,
    f22/13/2,
    f23/13/3,
    f24/13/4,
    lw1/-1/1,
    lw2/-1/-1,
    rw1/24/1,
    rw2/24/-1
    }
    \node[Vertex] (\e) at (\x,\y) {};

    \foreach \a/\b in 
    {{0/1},{1/2},{3/4},{4/5},{6/7},{7/8},{9/10},{10/11},{12/13},{13/14},
    {15/16},{16/17},{18/19},{19/20},{21/22},{22/23},
    {1/p11},{p11/p12},{1/pk1},{pk1/pk2},{p1l1/p1l2},{pkl1/pkl2},
    {1/p21},{p21/p22},{p2l1/p2l2}, 
    {7/f11},{f11/f12},{f13/f14},{13/f21},{f21/f22},{f23/f24},
    {0/lw1},{0/lw2},{23/rw1},{23/rw2}}
    \draw[Edge](\a) to node {} (\b);

    \node[Dots] () at (0,2.6) {$\vdots$};
    \node[Dots] () at (2,2.6) {$\vdots$};
    \node[Dots] () at (1,2.6) {$\vdots$};
    \node[Dots] () at (7,2.6) {$\vdots$};
    \node[Dots] () at (13,2.6) {$\vdots$};
    \node[Dots] () at (2.55, 0) {$\dots$};
    \node[Dots] () at (5.55, 0) {$\dots$};
    \node[Dots] () at (8.55, 0) {$\dots$};
    \node[Dots] () at (11.55, 0) {$\dots$};
    \node[Dots] () at (14.55, 0) {$\dots$};
    \node[Dots] () at (17.55, 0) {$\dots$};
    \node[Dots] () at (20.55, 0) {$\dots$};
    \node[Dots] () at (-1, 0.1) {$\vdots$};
    \node[Dots] () at (24, 0.1) {$\vdots$};

    \draw [decorate,decoration={brace,amplitude=7pt},xshift=0pt,yshift=0pt]
    (-0.2,4.3) -- (2.2,4.3) node [black,midway,yshift=18pt] {$k$ threads};

    \draw [decorate,decoration={brace,amplitude=7pt},xshift=0pt,yshift=0pt]
    (-1.3,-1.2) -- (-1.3,1.2) node [black,midway,xshift=-30pt] {$2b-1$};
    
    \draw [decorate,decoration={brace,amplitude=7pt},xshift=0pt,yshift=0pt]
    (24.3,1.2) -- (24.3,-1.2) node [black,midway,xshift=30pt] {$2b-1$};
    
    \draw [decorate,decoration={brace,amplitude=7pt},xshift=0pt,yshift=0pt]
    (1,-1.5) -- (-1.5,-1.5) node [black,midway,yshift=-17pt] {First wall};
    
    \draw [decorate,decoration={brace,amplitude=7pt},xshift=0pt,yshift=0pt]
    (4,-1.5) -- (1,-1.5) node [black,midway,yshift=-30pt] {First wasteland};
    
    \draw [decorate,decoration={brace,amplitude=7pt},xshift=0pt,yshift=0pt]
    (7,-1.5) -- (4,-1.5) node [black,midway,yshift=-17pt] {First gateland};
    
    \draw [decorate,decoration={brace,amplitude=7pt},xshift=0pt,yshift=0pt]
    (10,-1.5) -- (7,-1.5) node [black,midway,yshift=-30pt] {Selector};
    
    \draw [decorate,decoration={brace,amplitude=7pt},xshift=0pt,yshift=0pt]
    (13,-1.5) -- (10,-1.5) node [black,midway,yshift=-17pt] {Middle gateland};
    
    \draw [decorate,decoration={brace,amplitude=7pt},xshift=0pt,yshift=0pt]
    (16,-1.5) -- (13,-1.5) node [black,midway,yshift=-30pt] {Validator};
    
    \draw [decorate,decoration={brace,amplitude=7pt},xshift=0pt,yshift=0pt]
    (19,-1.5) -- (16,-1.5) node [black,midway,yshift=-17pt] {Last gateland};
    
    \draw [decorate,decoration={brace,amplitude=7pt},xshift=0pt,yshift=0pt]
    (22,-1.5) -- (19,-1.5) node [black,midway,yshift=-30pt] {Last wasteland};
    
    \draw [decorate,decoration={brace,amplitude=7pt},xshift=0pt,yshift=0pt]
    (24.2,-1.5) -- (22,-1.5) node [black,midway,yshift=-17pt] {Last wall};

    \node (l1) at (7,5.5) {First filler}; 
    \node (l2) at (13,5.5) {Second filler}; 
  \end{tikzpicture}
  \caption{A subgraph of $T$ with the components marked.}
  \label{figures:main-path-sections}
\end{figure}

When discussing vertices and subgraphs of $T$ we will apply an ordering based
on the distance from the center of the first wall, the leftmost wall in
Figure~\ref{figures:main-path-sections}. We will say that a vertex $u$ comes
before a vertex $v$ if $u$ is closer to the center of the first wall than $v$.
For subgraphs, we will compare the minimized distance over all vertices in each
subgraph. To complete our construction we need an ordering of the vertices of
$G$, we therefore let $V(G) = \left\{ v_1, \dots, v_{|V(G)|} \right\}$.

\subsubsection*{The First Wall, Wasteland and Gateland}
To ensure enough space for the gadgets in the validator we introduce
the \emph{pull-factor} $p$, which will correspond to the distance from the
$\gin$ vertex of a hole in the selector to the $\gin$ vertex of the next hole.
The pull-factor is $4n+3$ in our reduction, but will for convenience mostly be
referred to as $p$.

The first sector we will embed is the first wall. This is done by turning 
$u_1$ into the center of a wall by attaching leafs to it. Second comes the
first wasteland. This is done by attaching nothing to the vertices $u_2$ until
$u_{m_1}$ for $m_1 = pnk+2$. Note that $u_2$ is the vertex for which the
threads are connected. After this we embed $bm_1$ consecutive $k$-gates from
$u_{m_1}$ to $u_{(2b+1)m_1}$ to create the first gateland. This is done in such
a way that the $\gin$ vertex of the $i$'th gate is the $\gout$ vertex of the
$i-1$'th gate.

\subsubsection*{The Selector}
The selector will control the choices done by the threads. The idea is to let
the selector have $|V(G)|$ sparse intervals, namely holes, and let each of the
threads have a big knot, which can only be placed within such an interval. The
vertex selected by a thread is then decided by which hole its knot is placed
within. 

The embedding of the selector starts where the first gateland ended, at vertex
$u_{(2b+1)m_1}$. Note that this is the vertex where the first filler is
attached in Figure~\ref{figures:reduction-frame}. We now embed $|V(G)|$ $(k+1)$
holes with $(p-3)/2$ consecutive $(k+1)$-gates in between every consecutive
pair of holes on the path $(u_{(2b+1)m_1}, \dots, u_{(2b+1)m_1+p(n-1)+3})$.
After this we embed $b(p(n-1)+3)$ consecutive $(k+1)$-gates. In total, the
selector is embedded on the vertices $(u_{(2b+1)m_1}, \dots, u_{m_2})$ for $m_2
= (2b+1)m_1 + (2b+1)(p(n-1)+3)$.

\begin{figure}[ht!]
  \centering
  \begin{tikzpicture}
    \tikzset{scale=0.5}
    \tikzset{Vertex/.style={shape=circle,draw,scale=0.7}}
    \tikzset{Edge/.style={}}
    \tikzset{Dots/.style={scale=0.4}}
    \tikzset{LBag/.style={shape=ellipse,draw,scale=0.5,minimum
    size=10pt,minimum width=25pt}}
    
    \foreach \e/\x/\y [count=\k] in {
    0/0/0,
    1/1/0,
    2/2/0,
    3/3/0,
    4/4/0,
    5/5/0,
    6/6/0,
    7/7/0,
    8/8/0,
    9/9/0,
    10/10/0,
    11/11/0,
    12/12/0,
    13/13/0,
    14/14/0,
    15/15/0,
    16/16/0,
    17/17/0,
    18/18/0,
    19/19/0,
    20/20/0,
    21/21/0,
    22/22/0,
    23/23/0,
    24/24/0,
    25/25/0,
    26/26/0, 
    27/27/0,
    28/28/0,
    29/29/0,
    30/30/0,
    p10/2/1,
    p11/2/2,
    p12/2/3,
    p13/2/4,
    p14/2/5,
    p15/3/6,
    p16/4/6,
    p17/5/6,
    p18/6/6,
    p19/7/6,
    p110/8/6,
    p111/9/6,
    p112/10/6,
    p113/11/6,
    p114/12/6,
    p115/13/6,
    p116/14/6,
    p117/15/6,
    p118/16/6,
    p119/17/6,
    p120/18/6,
    p121/19/6,
    p122/20/6,
    p123/21/6,
    p124/22/6,
    p125/23/6,
    p126/24/6,
    p127/25/6,
    p128/26/6,
    p129/27/6,
    p130/28/6,
    p20/1/1,
    p21/1/2,
    p22/1/3,
    p23/1/4,
    p24/1/5,
    p25/1/6,
    p26/1/7,
    p27/2/8,
    p28/3/8,
    p29/4/8,
    p210/5/8,
    p211/6/8,
    p212/7/8,
    p213/8/8,
    p214/9/8,
    p215/10/8,
    p216/11/8,
    p217/12/8,
    p218/13/8,
    p219/14/8,
    p220/15/8,
    p221/16/8,
    p222/17/8,
    p223/18/8,
    p224/19/8,
    p225/20/8,
    p226/21/8,
    p227/22/8,
    p228/23/8,
    p229/24/8,
    p230/25/8,
    p231/26/8,
    p232/27/8,
    p233/28/8,
    pk0/0/1,
    pk1/0/2,
    pk2/0/3,
    pk3/0/4,
    pk4/0/5,
    pk5/0/6,
    pk6/0/7,
    pk7/0/8,
    pk8/0/9,
    pk9/1/10,
    pk10/2/10,
    pk11/3/10,
    pk12/4/10,
    pk13/5/10,
    pk14/6/10,
    pk15/7/10,
    pk16/8/10,
    pk17/9/10,
    pk18/10/10,
    pk19/11/10,
    pk20/12/10,
    pk21/13/10,
    pk22/14/10,
    pk23/15/10,
    pk24/16/10,
    pk25/17/10,
    pk26/18/10,
    pk27/19/10,
    pk28/20/10,
    pk29/21/10,
    pk30/22/10,
    pk31/23/10,
    pk32/24/10,
    pk33/25/10,
    pk34/26/10,
    pk35/27/10,
    pk36/28/10,
    lw1/-1/1,
    lw2/-1/-1,
    f1/12/1,
    f2/12/2,
    f3/12/3,
    f4/13/4,
    f5/14/4,
    f6/15/4,
    f7/16/4,
    f8/17/4,
    f9/18/4,
    f10/19/4,
    f11/20/4,
    f12/21/4,
    f13/22/4,
    f14/23/4,
    f15/24/4,
    f16/25/4,
    f17/26/4,
    f18/27/4
    }
    \node[Vertex] (\e) at (\x,\y) {};

    \node[LBag] (lb1) at (6,1) {$\gamma$};
    \node[LBag] (lb2) at (8,1) {$\gamma$};
    \node[LBag] (lb3) at (11,1) {$\gamma$};
    \node[LBag] (lb4) at (13,1) {$\eta$};
    \node[LBag] (lb5) at (14,1) {$\eta$};
    \node[LBag] (lb6) at (16,1) {$\Gamma$};
    \node[LBag] (lb7) at (19,1) {$\Gamma$};
    \node[LBag] (lb8) at (21,1) {$\eta$};
    \node[LBag] (lb9) at (22,1) {$\eta$};
    \node[LBag] (lb12) at (29,1) {$\gamma$};
    \node[LBag] (lb10) at (24,1) {$\Gamma$};
    \node[LBag] (lb11) at (27,1) {$\Gamma$};
    \node[LBag] (lb14) at (21.5,5) {$\kappa$};
    \node[LBag] (lb15) at (21.5,7) {$k$};
    \node[LBag] (lb16) at (21.5,9) {$k$};
    
    \foreach \a/\b in 
    {{0/1},{1/2},{2/3},{4/5},{1/p10},
    {1/p20},{1/pk0},
    {5/6},{6/7},{7/8},{8/9},{0/lw1},{0/lw2},{lb1/6},{lb2/8},{10/11},{11/12},{11/lb3},{12/13},
    {13/14},{lb4/13},{lb5/14},{14/15},{15/16},{16/lb6},{16/17},{18/19},
    {19/20},{20/21},{21/22},{22/23},{23/24},{24/25},{26/27},{27/28},{28/29},{29/30},{lb7/19},{lb8/21},
    {lb9/22},{lb12/29},{lb10/24},{lb11/27},{12/f1},{f1/f2},{f2/f3},{f4/f5},{f5/f6},{f6/f7},{f7/f8},{f8/f9},
    {f9/f10},{f11/f12},{f12/f13},{f13/f14},{f14/f15},{f15/f16},{f16/f17},{f17/f18},
    {p10/p11},{p11/p12},{p12/p13},{p13/p14},{p15/p16},{p16/p17},{p17/p18},{p18/p19},{p19/p110},{p110/p111},
    {p112/p113},{p113/p114},{p114/p115},{p115/p116},{p116/p117},{p117/p118},{p118/p119},{p119/p120},
    {p120/p121},{p121/p122},{p122/p123},{p123/p124},{p124/p125},{p125/p126},{p126/p127},{p127/p128},{p128/p129},{p129/p130},
    {p20/p21},{p21/p22},{p22/p23},{p23/p24},{p24/p25},{p25/p26},{p27/p28},{p28/p29},{p29/p210},{p210/p211},{p211/p212},
    {p212/p213},{p213/p214},{p215/p216},{p216/p217},{p217/p218},{p218/p219},{p219/p220},
    {p220/p221},{p221/p222},{p222/p223},{p223/p224},{p224/p225},{p225/p226},{p226/p227},{p227/p228},{p228/p229},{p229/p230},
    {p230/p231},{p231/p232},{p232/p233},
    {pk0/pk1},{pk1/pk2},{pk2/pk3},{pk3/pk4},{pk4/pk5},{pk5/pk6},{pk6/pk7},{pk7/pk8},{pk9/pk10},{pk10/pk11},{pk11/pk12},
    {pk12/pk13},{pk13/pk14},{pk14/pk15},{pk15/pk16},{pk16/pk17},{pk18/pk19},{pk19/pk20},
    {pk20/pk21},{pk21/pk22},{pk22/pk23},{pk23/pk24},{pk24/pk25},{pk25/pk26},{pk26/pk27},{pk27/pk28},{pk28/pk29},{pk29/pk30},
    {pk30/pk31},{pk31/pk32},{pk32/pk33},{pk33/pk34},{pk34/pk35},{pk35/pk36}}
    \draw[Edge](\a) to node {} (\b);

    \draw[Edge, bend left](f3) to node {} (f4);
    \draw[Edge, bend left](p14) to node {} (p15);
    \draw[Edge, bend left](p26) to node {} (p27);
    \draw[Edge, bend left](pk8) to node {} (pk9);
    \draw[Edge, bend left](p123) to node {} (lb14);
    \draw[Edge, bend left](p226) to node {} (lb15);
    \draw[Edge, bend left](pk29) to node {} (lb16);

    \node[Dots] () at (-1, 0.1) {$\vdots$};
    \node[Dots] () at (3.55, 0) {$\dots$};
    \node[Dots] () at (9.55, 0) {$\dots$};
    \node[Dots] () at (17.55, 0) {$\dots$};
    \node[Dots] () at (25.55, 0) {$\dots$};
    \node[Dots] () at (30.55, 0) {$\dots$};
    \node[Dots] () at (19.55, 4) {$\dots$};
    \node[Dots] () at (28.55, 6) {$\dots$};
    \node[Dots] () at (28.55, 8) {$\dots$};
    \node[Dots] () at (28.55, 10) {$\dots$};
    \node[Dots] () at (9.55, 6) {$\dots$};
    \node[Dots] () at (9.55, 8) {$\dots$};
    \node[Dots] () at (9.55, 10) {$\dots$};

    \draw [decorate,decoration={brace,amplitude=10pt},xshift=0pt,yshift=0pt]
    (5,-0.5) -- (0,-0.5) node [black,midway,yshift=-19pt] {First wasteland};
    
    \draw [decorate,decoration={brace,amplitude=10pt},xshift=0pt,yshift=0pt]
    (12,-0.5) -- (5,-0.5) node [black,midway,yshift=-20pt] {First gateland};

    \draw [decorate,decoration={brace,amplitude=10pt},xshift=0pt,yshift=0pt]
    (28,-0.5) -- (12,-0.5) node [black,midway,yshift=-20pt] {Selector};
  \end{tikzpicture}

    \caption{Illustration of the selector where $\gamma = 2(b-k-1)$, $\Gamma = 2(b-k-2)$,
      $\eta = \frac{3}{4}b-k-1$ and $\kappa =
      \frac{3}{2}b-k-1$.}
  \label{figures:chooser}
\end{figure}
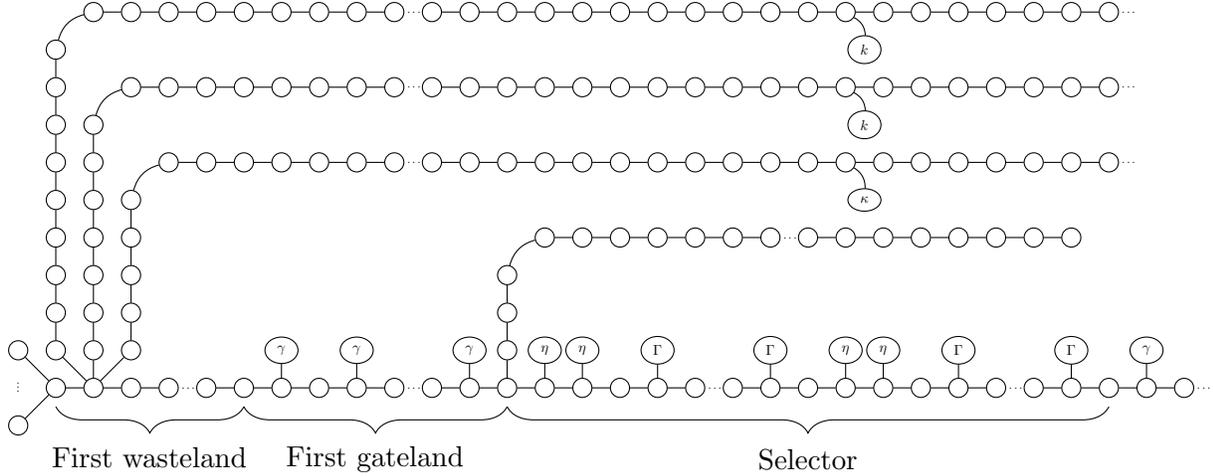

\subsubsection*{The Middle Gateland}
The middle gateland consist of $bm_2$ consecutive $k$-gates, embedded on the
main path from vertex $u_{m_2}$ to vertex $u_{(2b+1)m_2}$.

\subsubsection*{The Validator}
We will now give the validator. Its job is to verify that the selected vertices
of the threads in fact is a clique. The validator starts with $n-1$ neutral
zones, followed by a validation zone and another $n-1$ neutral zones. After
this there will be $b(2n-1)(4n+3)$ consecutive $(k+1)$-gates. A neutral zone is
a $P_{4n+4}$. The zones will be joined by sharing endpoints in the same style
as the gadgets in the selector. The validation zone consists of a $\path{4n+4}$
where there is $n$ $(k+1)$-holes sharing endpoints embedded on the last $3n+1$
vertices. The validator is hence embedded on the vertices $(u_{(2b+1)m_2},
\dots, u_{m_3})$ for $m_3 = (2b+1)m_2 + (2b+1)(2n-1)(4n+3)$.

\subsubsection*{The Last Gateland, Wasteland and Wall}
The last gateland consists of $bm_3$ consecutive $k$-gates embedded on the
vertices $(u_{m_3}, \dots, u_{(2b+1)m_3})$. After this we embed the last
wasteland, which means that we leave the vertices $(u_{(2b+1)m_3},
u_{b^2(2b+1)m_3})$ untouched. Finally we turn the vertex
$u_{b^2(2b+1)m_3+1}$ into the center of the last wall by attaching leafs to it.

\subsubsection*{The Threads and Their Danglements}

We will now describe the threads and their danglements. As they are all
isomorphic, it is sufficient to describe one of them. Let us name the vertices
on the thread by $t_2, \dots$ with $t_2 = u_2$. The leafs neighbouring to the
thread will be referred to as its danglements. First we turn $t_{(2b+1)m_1+1}$
into the center of a $(k+1)$-knot by attaching leafs to it. Starting at vertex
$t_{ (2b+1)m_2 + (n-1)(4n+3)}$ we consider $n$ consecutive, disjoint
$P_{4n+3}$. For $\path{4n+3}$ number $i$ we do the following. We divide the
$\path{4n+3}$ into disjoint subpaths, first a $\path{n+3}$ followed by $n$
$P_3$.  Consider the $j$'th $P_3$. If $i = j$ we turn the middle vertex of the
$P_3$ into the center of a $(k+1)$-knot.  Otherwise we attach a single leaf to
the middle vertex if $(v_i,v_j) \notin E(G)$. This leaf will be referred to as
a non-neighbouring leaf. After this we extend the thread with additional
$b(2b+1)m_3$ vertices. We would like to make the reader aware of the fact that
a thread is a path and the vertices connected to it, is its dangelments.

\subsubsection*{The Fillers}
A fillers job is to fill up all available room within a component of $T$ to
force this part of the main path to be stretched. To accomplish this we let the
filler connected to $u_{(2b+1)m_1}$ be of length $(n-k)(\frac{3}{2}b-k-2) +
(2b+1)(p(n-1)+3)$.  And the filler attached to $u_{(2b+1)m_2}$ to be
of length $(b-1)(4n+3)(2n-1) + 2b(4n+3)(2n-1) -(k(2n-1)(4n+3)  +
k(n(\frac{3}{2}b-k-2)+n^2-n-2m)+2n(\frac{3}{4}b-k-2))$.

\subsection{Correctness}
With the next lemmas we will prove the correctness of the reduction. After
this we will continue by giving the implications of this reduction, which are
the main results of this section. Recall that $b = 4k+16$ and $p = 4n+3$. 
\begin{lemma}
    \label{lemma:argangement}
    Given a yes-instance $(G, k)$ of \prob{$p$-Even Clique} the reduction
    instance $(T,b)$ is a yes-instance of \prob{$p$-Bandwidth}.
\end{lemma}

\begin{proof}

    We will now give a sparse ordering $\alpha$ of bandwidth $b=4k+16$, meaning
    that the image of $\alpha$ might not be an interval. To obtain a proper
    bandwidth ordering one can just compress $\alpha$. During the description
    of $\alpha$ a \emph{position} is a number in $\mathcal{N}$ that will be in
    the image of $\alpha$ and a vertex $v$ is said to be \emph{positioned} if
    the value $\alpha(v)$ has been given. Furthermore, we will say that $v$ is
    \emph{positioned at} $c$ if $\alpha(v) = c$. By \emph{reserving} a position
    for a subgraph $H$ of $T$ we guarantee that if a vertex will be positioned
    at that specific position, it will be a vertex of $H$. And by a position
    being \emph{available} we will mean that no vertex has been positioned at
    that specific position so far. Let $C_k = \left\{ c_1, \dots, c_k \right\}$
    be a $k$-clique in $G$. 
    
    For a vertex $u_i$ on the main path let $\alpha(u_i) = bi+1$. We continue
    by positioning the remainders of the two walls. And let $c_f$ be the center
    of the first wall and $L_f$ be the neighbouring leaves of $c_f$.  Let
    $\alpha(L_f) = \left[ \alpha(c_f)-b, \alpha(c_f)+b-1 \right] \setminus
    \left\{ \alpha(c_f) \right\}$ in some arbitrary way. Similarly for the
    last wall, let $\alpha(L_l) = \left[ \alpha(c_l)-b-1, \alpha(c_l)+b
    \right] \setminus \left\{ \alpha(c_l) \right\}$. Observe that for every
    two vertices $u$ and $v$ of $T$ such that both $\alpha(u)$ and $\alpha(v)$
    has been described, it holds that $\alpha(u) \neq \alpha(v)$. Furthermore,
    if $uv$ is an edge in $T$ it is true that $|\alpha(u) - \alpha(v)| \leq
    b$.

    Order the threads of $T$ and name them $\tau_1, \dots, \tau_k$.  Let $u$
    and $v$ be two neighbours on the main path such that neither $u$ nor $v$
    is the center of a wall and so that $\alpha(u) < \alpha(v)$.  Observe that
    there is $b-1$ available positions within $I(u,v)$.  Reserve the $k$
    positions in the middle of $I(u,v)$, one for each of the $k$ threads. If
    there are two positions equally close to the middle, take the leftmost
    one.  The leftmost is reserved for the first thread, the second to
    leftmost for the second thread and so forth.
    
    For every $i$ let $j_i$ be such that $c_i = v_{j_i}$. Consider hole number
    $j_i$ on the main path starting at the first wall, with $h_1, h_2, h_3$
    and $h_4$ being the vertices on the main path for which the hole is
    embedded on such that $\alpha(h_1) < \alpha(h_2) < \alpha(h_3) <
    \alpha(h_4)$. Thus $h_1, h_2, h_3$ and $h_4$ are the $\gin, \icenter,
    \ocenter$ and $\gout$ vertices of the hole respectively. Let $c$ be the
    center of the first knot on $\tau_i$ and $r$ the reserved position for
    $\tau_i$ in $\alpha$ within $I(h_2, h_3)$.  We then set $\alpha(c) = r$
    and complete the following procedure in the left (and right) direction on
    the thread $\tau_i$. Let $P$ be the path from $c$ to $N(u_2) \cap \tau_i$
    (or to the end of the thread). If every vertex of $P$ is positioned we
    stop. Otherwise, let $u$ be the vertex closest to $c$ on $P$ not yet
    positioned. Furthermore, let $\hat{P}_2$ be the rightmost (leftmost) $P_2$
    on the main path to the left (right) of the hole such that the position
    reserved for $\tau_i$ is available in $I(\hat{P}_2)$.  If $\hat{P}_2$ is
    not part of any wasteland we set $\alpha(u)$ to this reserved position and
    continue.  Otherwise we consider two cases. If we are right of $r$ we
    position $u$ at the leftmost position within $I(\hat{P}_2)$ that is either
    not reserved yet, or reserved for $\tau_i$. If we are left of $r$ we again
    consider two cases. Either there are exactly as many positions to the left
    of $r$ reserved for $\tau_i$ as there are vertices before $c$ not yet
    positioned.  In that case we position $u$ at the reserved position for
    $\tau_i$ within $I(\hat{P}_2)$.  Otherwise, we position $u$ at the
    rightmost position in $I(\hat{P}_2)$ that is either not reserved yet, or
    reserved for $\tau_i$. Observe that if $uv$ is an edge of
    $\tau_i$ there are positions reserved for $\tau_i$, $x$ and $y$, such that
    $y > x$ and $y-x = b$ and $\alpha(u)$ and $\alpha(v)$ are contained in
    $[x, y]$. It follows that $|\alpha(u) - \alpha(v)| \leq b$.
    
    Note that the number of vertices on a thread that will be positioned to the
    left of $r$ is $(2b+1)m_1$ and that, by construction, $2bm_1$ of these will
    be within the inclusion interval of the first gateland. Hence it can be
    observed that there are at most $km_1$ vertices from the threads within the
    inclusion interval of the first wasteland. Recall that the distance from
    $u_2$ to the first vertex of the first gateland is $m_1 - 2$. Hence there
    are $(b-1)(m_1-2) > km_1$ available positions within the inclusion interval
    of the first wasteland, before we position the threads.    By the same kind
    of argument there are $(b-1)(b^2-1)(2b+1)m_3$ available positions in the
    inclusion interval of the last wasteland before positioning the thread.
    Recall that the length of a thread is bounded above by
    \begin{align*} 
        & (2b+1)m_2+(n-1)(4n+3) + (4n+3)n + b(2b+1)m_3 \\ 
        <& (2b+3)m_2 + b(2b+1)m_3 \\
        <& 2b(2b+3)m_3.
    \end{align*}
    It follows that for every pair of vertices $u$ and $v$ such that both
    $\alpha(u)$ and $\alpha(v)$ has been described if holds that $\alpha(u)
    \neq \alpha(v)$.

    Recall that every $k$-gate of $T$ is embedded on the main path. And hence
    for every $k$-gate in $T$ there are $k$ paths passing through it with
    respect to $\alpha$. Hence there are $2(b-k-1)$ positions available
    between the left and the right leaf and the rest of the leaves can be
    positioned in any way within this interval. Clearly, for every pair of
    vertices $u$ and $v$ of $T$, such that both $\alpha(u)$ and $\alpha(v)$
    are described it holds that $\alpha(u) \neq \alpha(v)$. And furthermore,
    if $uv$ is an edge of a $k$-gate it holds that $|\alpha(u) - \alpha(v)|
    \leq b$.  For every $\path{2}$ on the main path such that $\path{2}$ is
    not in a subgraph of a wasteland and there are available positions in
    $I(\path{2})$ we reserve the position to the right of the $k$ positions
    reserved for the threads, for the fillers. Observe that any $\path{2}$
    such that this position is not available either is a subgraph of a
    wasteland or a $k$-gate (which has no available positions).

    We will now position the leaves of the  knots. Let $K$ be a knot in $T$.
    The center $c$ of $k$ is a vertex of a thread and hence $\alpha(c)$ has
    already been described.  Let $\hat{P}_2$ be the $P_2$ of the main path
    such that $\alpha(c) \in I(\hat{P}_2)$.  Position the leaves attached to
    $c$ as close to the middle of $I(\hat{P}_2)$ as possible by only using
    available positions, that are not reserved. If there are two such
    positions equally close to the middle, we take the leftmost one. Let
    $\hat{P}_4$ be the $P_4$ of the main path such that $\hat{P}_2$ contains
    the internal vertices of $\hat{P}_4$. It can be observed, by where the
    knots are embedded on the thread and where the threads are positioned in
    $\alpha$, that $\hat{P}_4$ is either a subgraph of a hole or a neutral
    zone. Furthermore, if $c'$ is the center of some other knot and
    $\hat{P}'_2$ is the $P_2$ of the main main such that $\alpha(c')$ is
    contained in its inclusion interval, then it can be observed that
    $\hat{P}'_2$ and $\hat{P}_4$ are disjoint. Hence, we see that there are
    $2(b-k-2)$ positions available and non-reserved within the inclusion
    interval of $\hat{P}_ 4$. Recall that a knot consists of
    $\frac{3}{2}b-k-2$ leaves and that $b = 4k+16$, and hence $2(b-k-2) \geq
    \frac{3}{2}b-k-2$. It follows that for every two vertices $u$ and $v$ of
    $T$ such that both $\alpha(u)$ and $\alpha(v)$ have been described,
    $\alpha(u) \neq \alpha(v)$. Furthermore, it $uv$ is an edge of $T$ it
    holds that $|\alpha(u) - \alpha(v)| \leq b$.

    Let $\path{4} = (u_h, u_{h+1}, u_{h+2}, u_{h+3})$ be some subpath of the
    main path such that a hole is embedded on it. Position the leaves attached
    to $u_{h+1}$ to the leftmost non-reserved, available positions and the
    leaves attached to $u_{h+2}$ to the rightmost non-reserved, available
    positions, within $I(\path{4})$. Furthermore, for the leaves representing
    non-neighbours, position it at the position available and not reserved
    closest to its neighbour. If there are two such positions, any of the two
    will do. It can be observed, by where the knots are embedded on the
    threads and where the knots and positioned that no two knots are
    positioned within the inclusion interval of a hole. And furthermore, that
    at most $k$ non-adjacency leaves are positioned within the inclusion
    interval of a hole. At last, since $C_k$ is a clique it holds that no knot
    and non-neighbour leaf is positioned with the inclusion interval of a
    hole. Recall that $k$ is even and hence $\frac{3}{2}b-k-2 = 5k+22$ is
    even. It follows that the leaves of a knot is evenly distributed among the
    two sides of the center. Recall that the number of leaves in a hole is
    $\frac{3}{2}b-2k-4$. There are $3k$ vertices from the threads
    positioned within the inclusion interval of $\hat{P}_4$ and there are
    $3b-3k-6$ leaves attached to one hole and one knot. Since there are more
    than $k$ leaves attached to a knot, it can be observed that for any two
    vertices $u$ and $v$ such that at least $u$ or $v$ is positioned
    within the inclusion interval of $\hat{P}_4$ it holds that $\alpha(u) \neq
    \alpha(v)$. And furthermore, if $uv$ is an edge in $T$ it holds that
    $|\alpha(u) - \alpha(v)| \leq b$.

    Consider danglements positioned within the inclusion interval of a
    $\hat{P}_4$ that is a subgraph of a neutral zone. One can observe that
    there is at most $k$ non-neighbouring leaves and at most one clique
    positioned within the inclusion interval of $\hat{P}_4$. And hence the
    same argument as above can be applied to show that for every two vertices
    $u$ and $v$ of $T$ such that both $\alpha(u)$ and $\alpha(v)$ has been
    described, it holds that $\alpha(u) \neq \alpha(v)$. Furthermore, if $uv$
    is an edge of $T$ it is true that $|\alpha(u) - \alpha(v)| \leq b$.

    It remains to describe the positioning of each of the fillers. Let $u$ be
    the vertex on the filler closest to the main path not yet positioned and
    $r$ lowest value bigger than the $\alpha$-value of the intersection vertex
    between the filler and the main path that is not taken. Set $\alpha(u) = r$
    and continue. Recall that the length of the path where the selector is
    embedded is $(n-1)p+3+2b(p(n-1)+3)$, and hence there were $(b-1)(
    (n-1)p+3+2b(p(n-1)+3) )$ available positions within the inclusion interval
    of the selector after only the main path had been positioned.  Observe that
    the threads now occupies $k((n-1)p+3+2b(p(n-1)+3))$ of these positions, the
    $k+1$-gates $( (p-3)(n-1)/2+b(p(n-1)+3))2(b-k-2)$ of the positions, the
    knots $k(\frac{3}{2}b-k-2)$ positions, the holes $2n(\frac{3}{4}b-k-2)$
    positions and the filler $(n-k)(\frac{3}{2}b-k-2) + (2b+1)(p(n-1)+3)$. By
    substituting $p$ by $4n+3$ and $b$ by $4k+16$ one can verify that the
    vertices positioned equals the amount of positions available within the
    inclusion interval of the selector. The expression for the once available
    positions within the inclusion interval of the selector $S$ and the number of
    vertices now positioned within it, disregarding the main path, namely $X$, is given
    below.
    \begin{align*}
            S &=  k((n-1)p+3+2b(p(n-1)+3)) \\
            &\hspace{0.5cm}+ ( (p-3)(n-1)/2+b(p(n-1)+3))2(b-k-2)\\
            &\hspace{0.5cm}+ k(\frac{3}{2}b-k-2) + 2n(\frac{3}{4}b-k-2) \\
            &\hspace{0.5cm}+ (n-k)(\frac{3}{2}b-k-2) + (2b+1)(p(n-1)+3) \\
            &= (b-1)( (n-1)p+3+2b(p(n-1)+3) ) = X.
    \end{align*}
    It follows that for every two vertices $u$ and $v$ such that both
    $\alpha(u)$ and $\alpha(v)$ have been described, it holds that $\alpha(u)
    \neq \alpha(v)$. Recall that for every $\hat{P}_2$ that is a subgraph of
    the main path and the selector there was a position reserved for the
    fillers. And hence for every edge $uv$ of the first filler, there are
    positions reserved for the filler, $x$ and $y$ such that $y-x = b$ and
    $\alpha(u)$ and $\alpha(v)$ is contained within $[x,y]$. It follows
    directly that $|\alpha(u) - \alpha(v)| \leq b$.  For the second filler, we
    observe that there were $(b-1)(4n+2)(2n-1)$ available positions within the
    inclusion interval of the validator when only the main path had been
    positioned. And furthermore, now the $n$ holes occupies
    $2n(\frac{3}{4}b-k-2)$ of these positions, the threads $k(2n-1)(4n+2)$ of
    the positions, the knots $kn(\frac{3}{2}b-k-2)$ and the non-neighbouring
    leaves $k(n^2-n-2m)$. By a similar argument as for the first filler, one
    can prove that for every $u$ and $v$ of $T$ it holds that $\alpha(u) \neq
    \alpha(v)$ and if $uv$ is an edge of $T$ then $|\alpha(u) - \alpha(v)| \leq
    b$. This completes the description of $\alpha$ and the argument is
    complete.
\end{proof}

Given a reduced instance $(T,b)$ and a $b$-bandwidth ordering $\alpha$ we say
that a $k$-gate in $T$ is \emph{blocked} with respect to $\alpha$ if every
thread in $T$ pass through the gate.

\begin{lemma}
    \label{lemma:threads-passes}
    Let $(T, b)$ be the result of the reduction for some instance of
    \prob{$p$-Even Clique} and $\alpha$ a $b$-bandwidth ordering of $T$.
    Then every $k$-gate in $T$ is blocked with respect to $\alpha$.
\end{lemma}

\begin{proof}
    By Lemma~\ref{lemma:walls} we know that the first wall is either the
    leftmost or the rightmost elements of $\alpha$. Observe that every $k$-gate
    in $T$ is blocked with respect to $\alpha$ if and only if every $k$-gate in
    $T$ is blocked with respect to $\alpha$ reversed. Hence it is
    sufficient to prove that every $k$-gate is blocked when the first wall is
    the leftmost elements of $\alpha$.
  
    Assume for a contradiction that there is a $k$-gate $\Pi$ and a thread
    $\tau$ such that $\tau$ is not passing through $\Pi$. Let $P$ be the path
    from $u_2$ to the $\gout$ vertex of $\Pi$ and let $X = V(\tau)-u_2$. By
    Lemma~\ref{lemma:walls} we know that $\alpha(u_2) = \min \alpha(\tau)$ and
    that $\alpha(u_2) < \min \alpha(\Pi)$. It follows by the
    definition of passing through that $\max \alpha(\tau) \leq \max \alpha(\Pi)$ and
    hence $\alpha(X) \subseteq I(P)$. Recall that $|E(P)| \leq (2b+1)m_3-2$ and
    $|X| > b(2b+1)m_3$. It follows directly that $|I(P)| \leq b( (2b+1)m_3-2)
    +1 <
    b(2b+1)m_3 < |X|$ which is a contradiction.
\end{proof}

Recall that the main path of the reduction instance consist of $9$ sectors,
namely the first wall, the first wasteland, the first gateland, the selector,
the middle gateland, the validator, the last gateland, the last wasteland and
the last wall. See Figure~\ref{figures:main-path-sections} for an
illustration. The lemma below shows that the sectors will appear in the same
order in $\alpha$ as they do in the instance, up to reversion.

\begin{lemma}
    \label{lemma:order-of-the-parts}
    Let $(T, b)$ be the result of the reduction for some instance of
    \prob{$p$-Even Clique} and $\alpha$ a $b$-bandwidth ordering of $T$ such
    that the first wall is mapped to the leftmost elements of $\alpha$. If $u$
    and $v$ are vertices from two different sectors such that $u$ comes before
    $v$ in $T$, then it holds that $\alpha(u) \leq \alpha(v)$.
\end{lemma}

\begin{proof}

    If at least one of the vertices are in one of the walls, the lemma follows
    directly from Lemma~\ref{lemma:walls}. We will now consider two cases.
    First, we consider the case when there is a $k$-gate $\Pi$ with center $c$
    embedded on the inner vertices of the path from $u$ to $v$. We  
    make $c$ adjacent to $\alpha^{-1}([\alpha(c)-b, \alpha(c)+b)$ and observe
    that $c$ is now the center of a wall and $\alpha$ is still a $b$-bandwidth
    ordering of the graph. Apply Lemma~\ref{lemma:walls} on the first wall and
    the new wall to obtain $\alpha(u) \leq \alpha(c)$ and on the new wall and
    the last wall to obtain $\alpha(c) \leq \alpha(v)$. It follows immediately
    that $\alpha(u) \leq \alpha(v)$.
    
    It remains to consider the case when there is no $k$-gate embedded on the
    inner vertices of the path from $u$ to $v$.  It follows, by construction,
    that either $u$ or $v$ is a vertex of a $k$-gate. First, let us consider
    the case when $u$ is a vertex of a $k$-gate. Recall that the vertices the
    gate is embedded on is named $\gin, c = \gcenter$ and $\gout$ and let $P$ be
    the path from $\gout$ to $v$. It follows by
    Lemmata~\ref{lemma:gates}~and~\ref{lemma:threads-passes} that $\alpha(P)$
    and $[\alpha(c) - b, \alpha(c)+b]$ intersects in only one element, namely
    $\alpha(\gout)$, and that $\alpha(\gin) < \alpha(c) < \alpha(\gout)$.
    Since $\alpha$ is a $b$-bandwidth ordering it follows that $\alpha(\gout)
    = \min \alpha(P)$ and hence $\alpha(u) \leq \alpha(\gout) \leq \alpha(v)$.
    The case when $v$ is a vertex of a $k$-gate follows by a symmetrical
    argument.
\end{proof}

Let $P_F, P_M$ and $P_L$ be the paths from the center of the first gate to the
center of the last gate in the first gateland, the middle gateland and the last
gateland respectively. 

\begin{lemma}
    \label{lemma:gatelands-are-stretched}
    Let $(T, b)$ be the result of the reduction for some instance of
    \prob{$p$-Even Clique} and $\alpha$ a $b$-bandwidth ordering of $T$,
    then 
    \begin{itemize}
        \item $P_F$, $P_M$ and $P_L$ are stretched with respect to $\alpha$ and
        \item for the centers of two $k$-gates $c_1$ and $c_2$ such that $c_1$
            comes before $c_2$ in $T$ it holds that $\alpha(c_1) <
            \alpha(c_2)$.
    \end{itemize}
\end{lemma}

\begin{proof}
    This follows directly from
    Lemmata~\ref{lemma:gates},~\ref{lemma:threads-passes}~and~\ref{lemma:order-of-the-parts}. 
\end{proof}

Let $\Pi_F$ and $\Pi_L$ be the first and last $k$-gate in $T$, and $c_F$ and
$c_L$ their centers respectively. Furthermore, let $P_R$ be the path from $c_F$ to $c_L$.

\begin{lemma}
    \label{lemma:positioning-of-high-degree-vertices-threads}
    Let $(T, b)$ be the result of the reduction for some instance of
    \prob{$p$-Even Clique} and $\alpha$ a $b$-bandwidth ordering of $T$.  If $u
    \neq u_2 $ is a vertex of a thread, such that the degree of $u$ is at least
    $3$, then $\alpha(u) \in I(P_R)$.
\end{lemma}

\begin{proof}
    By Lemma~\ref{lemma:walls} we know that the first wall is either the
    leftmost or the rightmost elements of $\alpha$. Observe that $u$ is mapped
    within the inclusion interval of $P_R$ by $\alpha$ if and only if $u$ is
    mapped within the inclusion interval of $P_R$ by $\alpha$ reversed. Hence
    it is sufficient to prove that $\alpha(u) \in I(P_R)$ when the first wall
    is the leftmost elements of $\alpha$.

    Assume for a contradiction that there is a vertex $u \neq u_2$ of some
    thread, such that $u$ has degree at least $3$ and $\alpha(u) \notin I(P_R)$.
    It follows from
    Lemmata~\ref{lemma:gatelands-are-stretched}~and~\ref{lemma:order-of-the-parts}
    that either $\alpha(u) < \alpha(c_F)$ or $\alpha(c_L) < \alpha(u)$. First,
    we consider the case when $\alpha(u) < \alpha(c_F)$. Let $P'$ be the path
    from $u_2$ to $u$ except $u_2$ and let $P''$ be the path from $u$ to the
    last vertex of the thread. Furthermore, let $P$ be the path from $u_2$ to
    $c_F$. Assume for a contradiction that there is a
    $k$-gate $\Pi$ such that $P''$ is not passing through $\Pi$. Let $P'$ be
    the path from $u_2$ to the $\gout$ vertex of $\Pi$. Observe that
    $\alpha(P'') \subseteq I(P')$. Recall that $|V(P'')| > b(2b+1)m_3$ and
    that $|E(P')| \leq b( (2b+1)m_3 - 2)$. It follows that $|E(P')| \leq b(
    (2b+1)m_3 - 2) < b(2b+1)m_3 < |V(P'')|$ and hence we get our contradiction.
    Hence $P''$ is passing through every $k$-gate. By
    Lemmata~\ref{lemma:gates}~and~\ref{lemma:threads-passes} we get that
    $\alpha(P') \subseteq I(P)$. Recall that $|V(P')| \geq (2b+1)m_1-1$ and
    that $|E(P)| = m_1$. It follows immediately that $|I(P)| \leq bm_1 + 1 <
    (2b+1)m_1-1\leq |V(P')|$ and hence we obtain a contradiction.

    It remains to consider the case when $\alpha(c_L) < \alpha(u)$. Let $P$ be
    the path from $u_2$ to $u$ and $P'$ the path from $u_2$ to $c_L$ except
    $u_2$. By assumption $\alpha(u_2) < \min \alpha(P')$ and hence $\alpha(P')
    \subseteq I(P)$. Recall that $|E(P)| < m_3$ and that $|V(P')| =
    (2b+1)m_3-3$. It follows that $|I(P)| < bm_3+1 < (2b+1)m_3-3 = |V(P')|$,
    which is a contradiction. 
\end{proof}

\begin{lemma}
    \label{lemma:stretched-reduction}
    Let $(T, b)$ be the result of the reduction for some instance of
    \prob{$p$-Even Clique} and $\alpha$ a $b$-bandwidth ordering of $T$. Then
    \begin{itemize}
        \item $|\alpha(\tau_i) \cap I(\path{2})| = 1$ for every thread
            $\tau_i$ and every subpath $\path{2}$ of $P_R$ and
        \item $P_R$ is stretched with respect to $\alpha$.
    \end{itemize}
\end{lemma}

\begin{proof}
    By Lemma~\ref{lemma:walls} we know that the first wall is either the
    leftmost or the rightmost elements of $\alpha$. Observe that $P_R$ is
    stretched with respect to $\alpha$ if and only if $P_R$ is stretched with
    respect to $\alpha$ reversed. It follows that it is sufficient to prove
    that the lemma holds when the first wall is the leftmost elements of
    $\alpha$.

    Let $Z = \alpha^{-1}(I(P_R))$ and observe that there are at most $2b$
    vertices in $N(Z)$. Furthermore, observe that every leaf of a gate or a
    hole is either within $I(P_R)$ or a neighbour of $Z$. It follows from
    Lemma~\ref{lemma:threads-passes} that Lemma~\ref{lemma:gates} applies to all
    $k$-gates of $T$. Furthermore, by Lemma~\ref{lemma:order-of-the-parts} it
    follows that the neighbours of the fillers are positioned after the first
    gateland and before the last gateland. And hence by
    Lemma~\ref{lemma:gates} and the fact that $\alpha$ is a $b$-bandwidth
    ordering, it follows that both fillers are positioned within $I(P_R)$. By
    Lemma~\ref{lemma:positioning-of-high-degree-vertices-threads} it holds
    that for every vertex $v$ that is a danglement, its neighbour is
    positioned within $I(P_R)$. And hence $v$ is either in $I(P_R)$ or a
    neighbour of $Z$.  Below you find a table giving an overview of how many
    vertices not on the main path, each type of gadget contributes with to $N[Z]$.

    \begin{center}
        \begin{tabular}[center]{|l|p{11cm}|}
            \hline
            Type of vertices & Amount \\
            \hline
            Knots & $k(n+1)(\frac{3}{2}b-k-2)$ \\
            \hline
            Holes & $4n(\frac{3}{4}b-k-2)$  \\
            \hline
            First filler &
            $(n-k)(\frac{3}{2}b-k-2)+(2b+1)(p(n-1)+3)$ \\
            \hline
            Second filler & $(b-1)(4n+3)(2n-1) + 2b(4n+3)(2n-1) -(k(2n-1)(4n+3)
            + k(n(\frac{3}{2}b-k-2)+n^2-n-2m)+2n(\frac{3}{4}b-k-2))$\\
            \hline
            $k$-gates & $2(b-k-1)b(m_1+m_2+m_3)$\\
            \hline
            $(k+1)$-gates & $2(b-k-2)((n-1)(p-3)+b(p(n-1)+3)+b(2n-1)(4n+3))$\\
            \hline
            non-neighbouring leafs & $k(n^2-n-2m)$\\
            \hline
        \end{tabular}
    \end{center}

    It follows from Lemma~\ref{lemma:walls} that there are two vertices of the
    main path within $N(Z)$. Let $X$ be all leaves in gates, holes and knots
    and non-neighbouring leaves and all the vertices in the fillers that are
    positioned within $I(P_R)$. We know that $|X|$ is at least the sum of the
    numbers in the table above, minus $2b-2$. And hence it can be verified
    that $|X| \geq (b-k-1)((2b+1)m_3-m_1-2)$. By construction it follows that
    $|E(P_R)| = (2b+1)m_3 - m_1-2$. And by
    Lemmata~\ref{lemma:threads-passes}~and~\ref{lemma:order-of-the-parts} it
    follows that all threads are passing through $P_R$ and hence we can apply
    Lemma~\ref{lemma:passing-paths-are-well-behaved} to complete the proof.
\end{proof}

Name the holes of the selector such that the first hole is called $H_1$ and the
last hole is $H_n$. Let $(T,b)$ be a resulting instance of the reduction and
$\alpha$ a $b$-bandwidth ordering of $T$. Furthermore, let $H_i$ be a hole of
$T$ embedded on the path $(v_1, v_2, v_3, v_4)$ such that $v_1$ comes before
$v_4$ in $T$. We say that a thread $\tau$ is \emph{selecting} $i$, if the
center $c$ of the first knot of the thread is positioned so that $\alpha(c) \in
I(v_2, v_3)$.

\begin{lemma}
    \label{lemma:selector-selects-uniquely}
    Let $(T, b)$ be the result of the reduction for the instance $(G,k)$ of
    \prob{$p$-Even Clique} and $\alpha$ a $b$-bandwidth ordering of $T$. Then
    every thread in $T$ selects a unique integer in $[n]$.
\end{lemma}

\begin{proof}
    By Lemma~\ref{lemma:walls} we know that the first wall is either the
    leftmost or the rightmost elements of $\alpha$. Observe that every thread
    in $T$ selects an unique integer with respect to $\alpha$ if and only if
    every thread in $T$ selects an unique integer with respect to $\alpha$
    reversed. It follows that it is sufficient to prove that the lemma holds
    when the first wall is the leftmost elements of $\alpha$.
    
    Let us consider a thread $\tau$ with vertices $(u_2 = t_2, t_3, \dots)$,
    where $c$ is the center of the first knot $K$ of $\tau$. Furthermore, let
    $c_F$ be the center of the first gate in the first gateland, $c_M$ the
    center of the last gate in the middle gateland and $c_L$ the center of the
    last gate in the last gateland. We will now prove that $\alpha(c) \in
    I(c_F, c_M)$. We know that $\alpha(c)
    \in I(P_R)$ by
    Lemma~\ref{lemma:positioning-of-high-degree-vertices-threads} and hence in
    $I(c_F, c_L)$ by Lemma~\ref{lemma:stretched-reduction}. Assume for a
    contradiction that $\alpha(c) \notin I(c_F, c_M)$, it follows that
    $\alpha(c) \in I(c_M, c_L)$. Let $P$ be the path from $u_2$ to $c$ and $P'$
    the path from $u_3$ to $c_M$. Observe that $\alpha(P') \subseteq I(P)$.
    Recall that $|E(P)| = (2b+1)m_1 - 1$ and that $V(P') = (2b+1)m_2-3$. A
    contradiction follows immediately, since $I(P) \leq b( (2b+1)m_1 - 1) + 1 <
    (2b+1)m_2-3 \leq V(P')$. And hence we can assume $\alpha(c) \in I(c_F,
    c_M)$.

    We will now prove that there is a hole $H_i$ such that $\alpha(c) \in
    I(H_i)$.  Assume for a contradiction that $\alpha(c) \notin
    I(H_i)$ for every $i$. Let $\hat{P}_2 = (p_1, p_2)$
    be the $P_2$ of the main path such that $\alpha(c) \in I(\hat{P}_2)$. It
    follows by construction, that either $p_1$ or $p_2$ is the center of a
    gate. Observe that the leaves attached to $c, p_1$ and $p_2$ must be
    positioned within a $\hat{P}_4$. And due to
    Lemma~\ref{lemma:stretched-reduction} there are $4+3k$ vertices from the
    main path and the threads within $I(\hat{P}_4)$.  Recall that there are
    $\frac{3}{2}b-k-2$ leaves attached to $c$ and at least $2(b-k-2)$ leaves
    attached to $\hat{P}_2$. This adds up to $4+3k+\frac{3}{2}b-k-2+2b-2k-4 =
    \frac{7}{2}b-2 > 3b+1$ and hence we get a contradiction.

    Let $H_i$ be embedded on the path $(v_1, v_2, v_3, v_4)$ such that $v_1$
    comes before $v_4$ in $T$.  Observe that due to
    Lemma~\ref{lemma:stretched-reduction} there is a position within the
    inclusion interval of the last $(k+1)$-gate of the selector that only the
    first filler can take. Due to our tight budget when it comes to positions
    within $I(P_R)$ (see the proof of Lemma~\ref{lemma:stretched-reduction})
    it follows that the first filler must take this position. And hence for
    every hole in the selector, the $(k+1)$-gate immediately before and after
    will be passed by the first filler. It follows that
    Lemma~\ref{lemma:gates} is applicable on the $(k+1)$-gates in the selector
    and hence $\alpha(K) \subseteq I(v_1, v_4)$.  Furthermore, due to
    Lemma~\ref{lemma:stretched-reduction} we know that $I(H_i) \subseteq
    I(v_1, v_4)$. And hence we can apply Lemma~\ref{lemma:r-hole} to obtain
    that $\alpha(c) \in I(v_2, v_3)$.

    It remains to prove that the threads selects unique integers. Assume
    otherwise for a contradiction and let $\tau$ and $\tau'$ be two threads
    selecting the same integer $i$. Hence there are two knots $K$ and $K'$
    such that $\alpha(K)\cup\alpha(K') \subseteq I(H_i)$. Observe that $I(H_i)
    = 3b+1 \geq 2(\frac{3}{2}b-k-2) + 2(\frac{3}{4}b-k-2) = 6b-4k-8 > 5b$ (since
    there are $\frac{3}{2}b-k-2$ leaves attached to a knot and
    $2(\frac{3}{4}b-k-2)$ leaves attached to a hole) and hence we get our
    contradiction and the proof is complete.
\end{proof}

\begin{lemma}
    \label{lemma:ordering-gives-clique}
    Let $(T,b)$ be the result of the reduction for the instance $(G,k)$ of
    \prob{$p$-Even Clique} and $\alpha$ a $b$-bandwidth ordering of $T$. Then
    the set $\left\{ v_i \mid \text{there is a thread selecting } i\right\}$
    is a clique in $G$.
\end{lemma}

\begin{proof}
    By Lemma~\ref{lemma:walls} we know that the first wall is either the
    leftmost or the rightmost elements of $\alpha$. Observe that the set of
    integers selected by the threads with respect to $\alpha$ is the same as
    the one selected with respect to $\alpha$ reversed. It follows that it is
    sufficient to prove that the lemma holds when the first wall is the
    leftmost elements of $\alpha$.
    
    Let $A$ be the set of selected integers and $C = \left\{ v_i \mid i \in A
    \right\}$.  From Lemma~\ref{lemma:selector-selects-uniquely} we know that
    the size of both $A$ and $C$ is $k$. Assume for a contradiction that there
    are two vertices $v_a$ and $v_b$ in $C$ such that $v_a$ and $v_b$ are not
    neighbours in $G$. Let $\tau_a$ be the thread selecting $a$ and $\tau_b$
    the thread selecting $b$. One can observe that by construction and
    Lemma~\ref{lemma:stretched-reduction} there is a hole $H$ in the
    validation zone and a knot $K_a$ with center $c_a$ embedded on $\tau_a$
    such that $\alpha(c_a) \in I(H)$.

    Let $(v_1, v_2,v_3,v_4)$ be the path that $H$ is embedded on, such that
    $v_1$ comes before $v_4$ in $T$. From
    Lemma~\ref{lemma:stretched-reduction} one can observe that there is a
    position within the inclusion interval of the last $(k+1)$-gate in the
    validator that only the second filler can take. Due to our tight budget
    when it comes to positions within $I(P_R)$ (see the proof of
    Lemma~\ref{lemma:stretched-reduction}) it follows that the second filler
    must take this position. It follows that Lemma~\ref{lemma:gates} is
    applicable on the $(k+1)$-gates immediately before and after $H$. Hence it
    follows by Lemma~\ref{lemma:r-hole} that $\alpha(K) \cup \alpha(H)
    \subseteq I(v_1, v_4)$. 

    From the construction of $T$ and Lemma~\ref{lemma:stretched-reduction} one
    can observe that the vertex of $\tau_b$ positioned within $I(v_2, v_3)$ has
    a non-neighbouring leaf attached. It follows that there are $3(k+1)+4$ vertices
    from the threads, the filler and the main path positioned within $I(v_1, v_4)$.
    Furthermore, the knot contributes with $\frac{3}{2}b-k-2$ leaves to $I(v_1,
    v_4)$ and the hole with $2(\frac{3}{3}b-k-2)$. And in addition the
    non-neighbouring leaf must be positioned within $I(v_1, v_4)$. It follows
    that $3b+1 = |I(v_1, v_4)| \leq 3(k+1)+4 + \frac{3}{2}b-k-2 +
    2(\frac{3}{4}b-k-2) + 1 = 3b + 7 - 2 - 4 + 1 = 3b+2$ which is a
    contradiction and the proof is complete.
\end{proof}

\begin{lemma}
    \label{lemma:correctnes}
    Given an instance $(G, k)$ of \prob{$p$-Clique} the reduction instance
    $(T,b)$ of \prob{$p$-Bandwidth} has a $b$-bandwidth ordering if and only if
    there is a clique of size $k$ in $G$.
\end{lemma}

\begin{proof}
    This follows immediately by
    Lemmata~\ref{lemma:argangement},~\ref{lemma:selector-selects-uniquely}~and~\ref{lemma:ordering-gives-clique}.
\end{proof}

\subsection{Consequences}

We will now present the immediate consequences of our reduction. But first we
need to prove that the problem we have been reducing from, namely
\prob{$p$-Even Clique} is up to the task.

\begin{lemma}
    \label{lemma:even-clique-w1hard}
    \prob{$p$-Even Clique} is \cW{1}-hard.
\end{lemma}

\begin{proof}
    We give a simple reduction from \prob{$p$-Clique}, which was proven to be
    \cW{1}-hard by Downey~\&~Fellows~\cite{downey1995fixed}. Given an
    instance $(G,k)$ of \prob{$p$-Clique}, if $k$ is even the instance is
    already a valid instance of \prob{$p$-Even Clique} and the correctness is
    trivial. Otherwise, let $G'$ be $G$ with a universal vertex added and $k' =
    k+1$. Clearly, $k'$ is even. So this is a valid instance. If there is a
    clique of size $k$ in $G$, then the same clique together with the universal
    vertex forms a clique of size $k'$ in $G'$. And the other way around, if
    there is a clique of size $k'$ in $G'$. Then there is a subset of this
    clique of size $k$ not containing the added universal vertex. This is a
    clique in $G$ of size $k$ and hence our reduction is sound.

    Since the reduction is parameter preserving it follows immediately that
    \prob{$p$-Even Clique} is \cW{1}-hard.
\end{proof}

\begin{lemma}
    \label{lemma:even-clique-eth}
    Assuming the Exponential Time Hypothesis \prob{$p$-Even Clique} does not
    admit an $O(f(b)n^{o(b)})$ time algorithm.
\end{lemma}

\begin{proof}
    Observe that for the reduction in the proof of
    Lemma~\ref{lemma:even-clique-w1hard} is so that $k' = O(k)$.
    \prob{$p$-Clique} is known to not admit an $O(f(b)n^{o(b)})$ time algorithm
    by Chen~et.~al.~\cite{chen2006strong}. The result follows immediately.
\end{proof}

\begin{theorem}
    \label{theorem:bandwidth-w1hard}
    \prob{$p$-Bandwidth} is \cW{1}-hard, even when the input
    graph is restricted to trees of pathwidth at most $2$.
\end{theorem}

\begin{proof}
    The result follows directly from
    Lemmata~\ref{lemma:correctnes}~and~\ref{lemma:even-clique-w1hard} and the
    observations that the graph constructed by the reduction is a tree of
    pathwidth at most 2 and that $b = f(k)$.
\end{proof}

\begin{theorem}
    \label{theorem:bandwidth-eth}
    Assuming the Exponential Time Hypothesis \prob{$p$-Bandwidth} does not
    admit an $O(f(b)n^{o(b)})$ time algorithm, even when the input graph is
    restricted to trees of pathwidth at most $2$.
\end{theorem}

\begin{proof}
The result follows directly from
    Lemmata~\ref{lemma:correctnes}~and~\ref{lemma:even-clique-eth} and the
    observations that the graph constructed by the reduction is a tree of
    pathwidth at most 2 and that $b = O(k)$.
\end{proof}

\newcommand{\alg}{\texttt{TreeAlg}}
\newcommand{\catalg}{\texttt{CatAlg}}
\newcommand{\cantoralg}{\texttt{FindSCC}}
\newcommand{\catar}{48b^3}
\newcommand{\catrun}{\textit{poly}}
\newcommand{\algar}{(768b^3)^b}
\newcommand{\algrun}{O(pbn^3)}
\newcommand{\false}{\bot}
\newcommand{\dist}{\mathrm{dist}}
\newcommand{\low}[1]{#1_{\mathrm{low}}}
\newcommand{\mbw}{\mathrm{mbw}}
\newcommand{\dir}{\mathrm{dir}}
\newcommand{\ldir}{\mathrm{W}}
\newcommand{\rdir}{\mathrm{E}}
\newcommand{\slen}{\mathrm{slen}}

\section{Approximation Algorithms}

In this section we will provide \cFPT-approximation algorithms for
\prob{$p$-Bandwidth} on trees and caterpillars. Given a caterpillar $T$ and a
positive integer $b$, $\catalg$ either returns a $\catar$-bandwidth ordering
of $T$ or correctly concludes that $bw(T) > b$. To obtain this we define an
obstruction for bandwidth on caterpillars inspired by
Chung~\&~Seymour~\cite{chung1989graphs} and search for these objects. Based on
the appearance of these objects in $T$ we construct an interval graph such
that either the interval graph has low chromatic number or the bandwidth of
$T$ is large.  If the interval graph has low chromatic number we use a
coloring of this graph to give a low bandwidth layout of $T$.

Given a tree $T$ and positive integers $b$ and $p$ such that $\pw(T) \leq p$,
$\alg$ either returns a $(768b^3)^p$-bandwidth ordering of $T$ or correctly
concludes that $bw(T) > b$. The high level outline of the algorithm is as
follows. The algorithm first decomposes the tree into several connected
components of smaller pathwidth and recurses on these. Then it builds a host
graph for $T$ that is a caterpillar, applies $\catalg$ on the host graph.
Finally it combines the result of $\catalg$ with the results from the
recursive calls, to give a $(768b^3)^p$-bandwidth ordering of $T$. Since the
pathwidth of a graph is known to be bounded above by its bandwidth, it follows
that $\alg$ is an \cFPT-approximation.

\subsection{An \cFPT-Approximation for the Bandwidth of Trees}
The aim of this section is to give a \cFPT-approximation for
\prob{$p$-Bandwidth} on trees, namely an $\algar$-approximation. This
algorithm crucially uses a $\catar$-approximation of \prob{$p$-Bandwidth} on
caterpillars as a subroutine. We provide such an algorithm, namely the
algorithm $\catalg$, in Section~\ref{sec:caterpillars}. In the remainder of
this section we give a $\algar$-approximation for trees under the assumption
that $\catalg$ is a $\catar$-approximation of \prob{$p$-Bandwidth} on
caterpillars with running time $O(bn^3)$.

\subsubsection*{Recursive Path Decompositions and Other Simplifications}

In this section we will present some decomposition results crucial for our
algorithm. First we define \emph{recursive path decompositions}, which will
allow us to partition our graph into several components of slightly lower
complexity.  The recursive decomposition is used to call the algorithm
recursively on easier instances, and then combine the layouts of these
instances to a low bandwidth layout of the input tree.

\begin{definition}
    Let $T$ be a tree and $P, T^1, \dots, T^t$ induced subgraphs of $T$ such
    that $V(T) = V(P) \cup \bigcup V(T^i)$. Then we say that $P, T^1,\dots,T^t$
    is a \emph{$p$-recursive path decomposition} of $T$ if $P$ is a path in $T$
    and for every $i$ it holds that $T^i$ is a connected component of $T - P$,
    $\deg(V(T^i)) = 1$ and $\pw(T^i) < p$.   
\end{definition}

\begin{lemma}
    \label{lemma:computing-recursive-path-decompositions}
    Given a tree $T$ of pathwidth at most $p$, a $p$-recursive path
    decomposition $P, T^1,\dots,T^t$ of $T$ can be found in $O(n)$ time.
\end{lemma}

\begin{proof}

    It was proven by Scheffler~\cite{scheffler1990linear} that given a tree
    $T$ and an integer $p$ one can find a path decomposition $\mathcal{P}$ of
    $T$ of width $p$ or correctly conclude that $\pw(T) > p$ in time $O(n)$.
    Let $X$ and $Y$ be the leaf bags of $\mathcal{P}$. By standard techniques
    we can assume $X$ and $Y$ to be non-empty. Let $u,v$ be two, not
    necessarily distinct, vertices such that $u \in X$ and $v \in Y$. Let $P$
    be the path in $T$ from $u$ to $v$. One can easily prove that for every
    bag $Z$ of $\mathcal{P}$ it is true that $Z \cap P$ is non-empty. Hence,
    if we remove all the vertices of $P$ from $T$ and $\mathcal{P}$ we obtain
    a path decomposition of $T-P$ of width $p-1$. It follows that for every
    connected component $T^i$ of $T-P$ it holds that $\pw(T^i) \leq p-1$.
    Assume for a contradiction that there is a connected component $T^i$ such
    that $\deg(V(T^i)) \neq 1$. If $\deg(V(T^i)) < 1$ it follows that $T$ was
    disconnected to begin with, and hence not a tree. And if $\deg(V(T^i)) >
    1$ it follows that $T^i$ together with $P$ forms a cycle, and again $T$ is
    not a tree. To complete the proof, observe that the connected components
    of $T-P$ can be found in $O(n)$ time by breadth first search.

\end{proof}

\begin{definition}
    Let $T$ be a tree and $P,T^1,\dots,T^t$ a $p$-recursive path decomposition
    of $T$. We construct the \emph{simplified instance} $T_S$ of $T$ with
    respect to $P,T^1,\dots,T^t$ as follows. First we add $P$ to $T_S$. Then,
    for every $T^i$ we first add a path $P^i$ such that $|V(P^i)| = |V(T^i)|$
    and then we add an edge from one endpoint of $P^i$ to $N(T^i)$.
\end{definition}
Observe that the simplified instance $T_S$ is a caterpillar with backbone $P$.
\begin{lemma}
    \label{lemma:simplified-instance-upper-bound}
    Let $T$ be a tree, $P, T^1, \dots, T^T$ be a $p$-recursive path
    decomposition of $T$ and $T_S$ the corresponding simplified instance, then
    $\bw(T_S) \leq 2\bw(T)$
\end{lemma}

\begin{proof}
    Let $\alpha$ be an optimal bandwidth ordering of
    $T$. We will now give an ordering $\beta$ of $T_S$ such that $\bw(T_S,
    \beta) \leq 2\bw(T, \alpha)$. For every $v \in P$, let $\beta(v) =
    2\alpha(v)$.
   
    For every $T^i$ we will consider two cases. Let $W = \alpha(T^i)$ and
    observe that for every $x \in W$ such that $y$ is the smallest element in
    $W$ larger than $x$ it follows by the connectivity of $T^i$ that $y-x \leq
    \bw(T)$. First, consider the case when at least half of $W$ is less than
    $\alpha(N(T^i))$. For every $w \in W$ such that $w < \alpha(N(T^i))$, add
    $2w$ and $2w+1$ to the initially empty set $Z$. Let $P^i = \left\{ p_1,
    \dots, p_m \right\}$ such that $\dist(P, p_j) < \dist(P, p_{j+1})$ for
    every $j$. For $j$ from $1$ to $m$, let $\beta(p_j)$ be the largest value
    in $Z$ and discard $\beta(p_j)$ from $Z$. Observe that for every $j$ it
    holds that $|\beta(p_j) - \beta(p_{j+1})|/2 \leq \bw(T)$. And furthermore,
    $|\beta(p_1) - \beta(N(P^i))| \leq \bw(T)$.  If at least half of $W$ is
    larger than $\alpha(N(T^i))$ apply a symmetric construction.

    To conclude the argument we need to prove that $\beta$ never maps two
    distinct vertices of $T_S$ on the same position. It is easy to verify that
    this never happens for two vertices on $P$ or two vertices in the same
    tree $T^i$.  Consider now a vertex $u \in V(T^i)$ and a vertex $v \in
    V(T^j)$ for $i \neq j$. It follows that $\lfloor \beta(u)/2 \rfloor \in
    \alpha(T^i)$ and $\lfloor \beta(v)/2 \rfloor \in \alpha(T^j)$. Since
    $\alpha(T^i) \cap \alpha(T^j) = \emptyset$ it follows that $\beta(u) \neq
    \beta(v)$. The argument for one vertex in $T^i$ and one in $P$ is
    identical. We obtain that $\bw(T_S) \leq \bw(T_S, \beta) \leq 2\bw(T,
    \alpha) = 2\bw(T)$.
\end{proof}

Let $T$ be a graph, $v$ a vertex of $T$ and $\alpha$ a $b$-bandwidth ordering
of $T$. Let $\beta'$ be a sparse ordering such that for every $u \in T$ 
\[\beta'(u) = 
    \begin{cases}
        2[\alpha(v)-\alpha(u)] & \mbox{if $\alpha(u) \leq \alpha(v)$ and} \\
        2[\alpha(u)-\alpha(v)]-1 & \mbox{otherwise.}
    \end{cases}\]
and let $\beta$ be the bandwidth ordering obtained by compressing $\beta'$. We
then say that $\beta$ is $\alpha$ \emph{right folded} around $v$. Observe
that $\bw(T, \beta) \leq 2\bw(T, \alpha)$.

\subsubsection*{Algorithm and Correctness}
We are now ready to describe algorithm $\alg$ and prove its correctness.
Pseudocode for $\alg$ is given in Algorithm~\ref{alg:alg}. 

\begin{algorithm}[h!]
  \KwIn{A tree $T$ and positive integers integers $p$ and $b$ such that $\pw(T) \leq
  p$.}
  \KwOut{A $(768b^3)^p$-bandwidth ordering of $T$ or conclusion that $\bw(T) > b$.}
\BlankLine
\If{$p = 1$}{
    \KwRet $\catalg(T, b)$ 
}
Find a $p$-recursive path decomposition $P, T^1, \dots, T^t$ of $T$.\\
Let $\alpha_1 = \alg(T^1, p-1, b), \dots, \alpha_t = \alg(T^t, p-1, b)$.\\
\If{there is an $\alpha_i = \false$} {
    \KwRet $\false$
}
Let $T_s$ be the simplified instance of $T$ with respect to $P,T^1,\dots,T^t$.\\
Let $\alpha_s = \catalg(T_s,2b)$.\\
\If{$\alpha_s = \false$} {
    \KwRet $\false$
}
For every $i$, let $\beta_i$ be $\alpha_i$ right folded around $N(P) \cap T^i$.\\
For every $v \in P$, let $\alpha(v) = \alpha_s(v)$.\\
For every $P_i$ of $T_s$ and every $v \in P_i$ of distance $d$ from $P$ in $T_s$, let
$\alpha(\beta_i^{-1}(d)) = \alpha_s(v)$.\\ 
\KwRet $\alpha$
  \caption{$\alg$}
  \label{alg:alg}
\end{algorithm}

\begin{lemma}
    \label{lemma:alg-running-time}
    Given a tree $T$ and two integers $p$ an $b$ such that $\pw(T) \leq p$,
    $\alg$ terminates in $\algrun$ time.
\end{lemma}

\begin{proof}
    We start by analyzing the time complexity of the computations done in a
    specific execution of $\alg$ given $T', p', b$ as input, disregarding the
    recursive calls. The calls to $\catalg$ require $O(b|V(T')|^3)$ time. Finding a
    $p$-recursive path decomposition can be done in $O(|V(T')|)$ time by
    Lemma~\ref{lemma:computing-recursive-path-decompositions}. Constructing $T'_S$
    can trivially be done in $O(|V(T')|)$ time. And furthermore, constructing all
    the $\beta$'s require $\sum_{i=1}^{t}{O(|T^i|)} = O(|V(T')|)$ time. Last, we
    observe that constructing $\alpha$ requires $O(|V(T')|)$ time. It follows that
    the time complexity of the computations done in a specific call to $\alg$
    is $O(b|V(T')|^3)$.
    
    Let $n = |V(T)|$ and $T_1,
    \dots, T_l$ the trees given as input at a specific recursion level. Observe
    that $T_1, \dots, T_l$ are pairwise disjoint and hence it follows that the
    time complexity of a recursion level is $\sum_{i=1}^{l}{O(b|V(T_1)|^3)} =
    O(bn^3)$. Furthermore, as $p$ is decreased by one at each recursion
    level it follows that $\alg$ runs in time $O(pbn^3)$.
\end{proof}

\begin{lemma}
    \label{lemma:alg-correctness}
    Given a tree $T$ and positive integers $b$ and $p$ such that $\pw(T) \leq
    p$, $\alg$ either returns a $O( (768b^3)^p)$-bandwidth ordering of $T$ or
    correctly concludes that $\bw(T) > b$ in time $\algrun$.
\end{lemma}

\begin{proof}
    The running time follows directly from Lemma~\ref{lemma:alg-running-time}
    and hence it remains to prove the correctness of the algorithm.  This we
    will do by induction on $p$. For $p=1$ the correctness follows directly
    from the correctness of $\catalg$ and hence it remains to prove the induction step.
    First we consider the case when the algorithm concluded that $\bw(T) > b$.
    Either there is an $\alpha_i$ such that $\alpha_i = \false$ or $\alpha_s =
    \false$. If $\alpha_i = \false$ it follows by the induction hypothesis and
    the fact that bandwidth is preserved on subgraphs that the algorithm
    concluded correctly. Now we consider the case when $\alpha_s = \false$. It
    follows from the correctness of $\catalg$ that $\bw(T_s) > 2b$ and hence by
    Lemma~\ref{lemma:simplified-instance-upper-bound} it follows that $\bw(T) >
    b$.
    
    It remains to consider the case when the algorithm returns a bandwidth
    ordering $\alpha$. Then, by the induction hypothesis $\alpha_i$ is a
    $(768b^3)^{p-1}$-bandwidth ordering of $T^i$ for every $i$. Furthermore,
    $\alpha_s$ is a $384b^3$-bandwidth ordering for $T_s$, since $48(2b)^3 =
    384b^3$. Let $u$ and $v$ be two neighbouring vertices of $T$. If $u$ and
    $v$ are vertices in $P$ it follows from $\bw(T_s, \alpha_s) \leq 384b^3$
    that $|\alpha(u) - \alpha(v)| \leq 384b^3$. Next, we consider the case when
    either $u$ or $v$ is a vertex in $P$. Assume without loss of generality
    that $u \in P$ and let $T^j$ be such that $v \in T^j$. By the definition
    of $\beta_j$ it follows that $\beta_j(v) = 1$. It follows that $|\alpha(u)
    - \alpha(v)| = |\alpha_s(u) - \alpha_s(w)|$ where $\dist(u,w) = 1$, and
    hence $u$ and $w$ are neighbours in $T_s$ and it follows directly that
    $|\alpha(u) - \alpha(v)| \leq 384b^3$. We will now consider the case when
    $u$ and $v$ are vertices of $T^j$ for some $j$. Let $u'$ be the vertex in
    $P^j$ of distance $\beta(u)$ from $P$ and $v'$ the vertex in
    $P^j$ of distance $\beta(v)$ from $P$. It follows that

    \begin{align*}
        |\alpha(u) - \alpha(v)| &= |\alpha(\beta_j^{-1}(\beta_j(u))) -
        \alpha(\beta_j^{-1}(\beta_j(v)))| \\
        &= |\alpha_s(u') - \alpha_s(v')| \\
        &\leq \dist(u',v')384b^3 \\ 
        &=|\beta_j(u) - \beta_j(v)|384b^3 \\
        &\leq |\alpha_j(u) - \alpha_j(v)| 768b^3 \\
        &\leq (768b^3)^p
    \end{align*}
    completing the proof.
\end{proof}

Note that one in the case of $p=1$ also could solve the instance exactly by
Assmann~\cite{assmann1981bandwidth}. It would decrease the approximation ratio
to $(768b^3)^{p-1}$.

\begin{theorem}
    \label{theorem:tree-alg}
    There exists an algorithm that given a tree $T$ and a positive integer $b$
    either returns a $(768b^3)^b$-bandwidth ordering of $T$ or correctly
    concludes that $\bw(T) > b$ in time $O(b^2n^3)$.
\end{theorem}

\begin{proof}
    This follows directly from $\pw(T) \leq \bw(T)$ and
    Lemma~\ref{lemma:alg-correctness}.
\end{proof}

The proof of Theorem~\ref{theorem:tree-alg} assumed the existence of a
$48b^3$-approximation algorithm for caterpillars. In the next section we give
such an algorithm.

\subsection{An \cFPT-Approximation for the Bandwidth of Caterpillars}
\label{sec:caterpillars}

The bandwidth of caterpillars is, somewhat surprisingly, a well-studied
problem. Assmann~et~al.~\cite{assmann1981bandwidth} proved that the bandwidth
of caterpillars of stray length $1$ and $2$ is polynomial time computable.
Monien~\cite{monien1986bandwidth} completed the story of polynomial time
computability by proving that \prob{Bandwidth} on caterpillars of stray length
$3$ is \cNP-hard. Furtermore,
Haralambides~\cite{haralambides1991bandwidth} gave an $O(\log n)$
approximation algorithm, which later was improved to $O(\log n/ \log \log n)$
by Feige~\&~Talwar~\cite{feige2009approximating}.  We now give the first
\cFPT-approximation of \prob{$p$-Bandwidth} on caterpillars, namely a
$\catar$-approximation.

\subsubsection*{Skewed Cantor Combs}

Chung~\&~Seymour~\cite{chung1989graphs} defined \emph{Cantor combs}. These are
very special caterpillars defined in such a way that they have small local density,
but high bandwidth. The definition of Cantor combs is very strict - it
precisely defines the length of all the paths in the caterpillars. For our
purposes we need a more general definition which captures all caterpillars
that are ``similar enough'' to Cantor combs. We call such caterpillars {\em
skewed Cantor combs}, and we will prove that they also have high
bandwidth.   Our algorithm will scan for skewed Cantor combs as an obstruction
for bandwidth and if none of big enough size are found it will construct a
$\catar$-bandwidth ordering based on the appearance of smaller versions of
these objects.  

For positive integers $k \leq b$ we now define a \emph{skewed $b$-Cantor comb}
of depth $k$, denoted $S_{b,k}$ inductively as follows. $S_{b,1}$ is a path of
length $1$. For the induction step to be well-defined we mark two vertices of
every skewed $b$-Cantor comb as end vertices. For an $S_{b,1}$ the two
vertices are the end vertices. For $k > 1$ we start with two skewed $b$-Cantor
combs of depth $k-1$, lets call them $S$ and $S'$ and furthermore let $x, y$
and $x', y'$ be their end vertices respectively. Connect $y$ to $x'$ by a path
$P$ of length at least $2$. Furthermore, let $Q$ be a stray connected to an
internal vertex $v$ of $P$. Mark $x$ and $y'$ as the end vertices of the
construction and let $B$ be the path from $x$ to $y'$. Let $d$ be the maximum
distance from $v$ to any vertex in $B$.  If $Q$ has at least $2(b-1)d$
vertices we say that the graph described is a skewed $b$-Cantor comb of depth
$k$.

\begin{figure}[ht!]
  \centering
  \begin{tikzpicture}
    \tikzset{scale=1}
    \tikzset{Vertex/.style={shape=circle,draw,scale=0.9}}
    \tikzset{Edge/.style={}}
    \tikzset{Dots/.style={scale=0.8}}

    \foreach \e/\x/\y [count=\k] in {
    0/0/0,
    1/1/0,
    2/2/0,
    3/3/0,
    4/4/0,
    5/5/0,
    6/6/0,
    7/7/0,
    8/8/0,
    9/9/0,
    10/10/0,
    s11/2/1,
    s12/2/2,
    q1/5/1,
    q2/5/2,
    q3/5/3,
    q4/5/4,
    q5/5/5,
    s21/8/1,
    s22/8/2
    }
    \node[Vertex] (\e) at (\x,\y) {};

    \foreach \a/\b in 
    {{0/1},{1/2},{2/3},{3/4},{6/7},{7/8},{8/9},{9/10},{5/q1},{q1/q2},
    {q3/q4},{q4/q5},{2/s11},{8/s21}}
    \draw[Edge](\a) to node {} (\b);

    \node[Dots] () at (4.5,0) {$\hdots$};
    \node[Dots] () at (5.5,0) {$\hdots$};
    \node[Dots] () at (2,1.6) {$\vdots$};
    \node[Dots] () at (5,2.6) {$\vdots$};
    \node[Dots] () at (8,1.6) {$\vdots$};

    \draw [decorate,decoration={brace,amplitude=7pt},xshift=0pt,yshift=0pt]
    (4,-0.25) -- (0,-0.25) node [black,midway,yshift=-18pt] {$S$};

    \draw [decorate,decoration={brace,amplitude=7pt},xshift=0pt,yshift=0pt]
    (6,-0.25) -- (4,-0.25) node [black,midway,yshift=-18pt] {$P$};

    \draw [decorate,decoration={brace,amplitude=7pt},xshift=0pt,yshift=0pt]
    (10,-0.25) -- (6,-0.25) node [black,midway,yshift=-18pt] {$S'$};

    \draw [decorate,decoration={brace,amplitude=7pt},xshift=0pt,yshift=0pt]
    (5.25,5) -- (5.25,1) node [black,midway,xshift=18pt] {$Q$};
  \end{tikzpicture}
  \caption{A skewed $b$-Cantor comb of depth $3$ for some $b$.}
  \label{figures:cantor}
\end{figure}
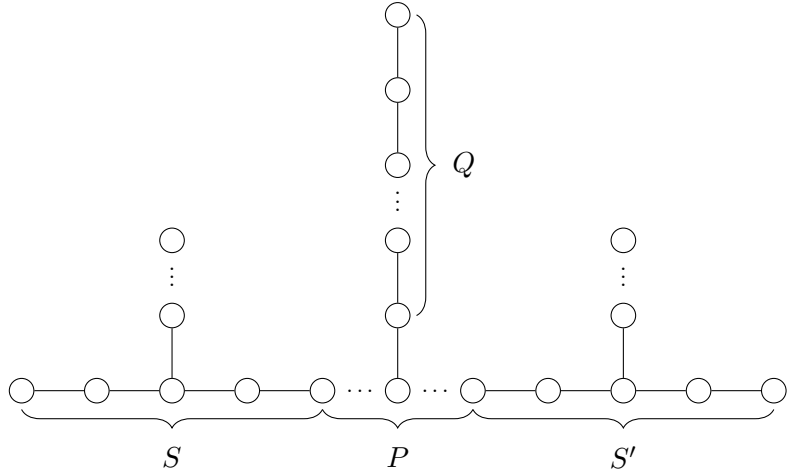

\begin{lemma}
    \label{lemma:skewed-cantor-combs-high-bandwidth-close-to-backbone}
    Let $\hat{S}_{b,k}$ be a skewed $b$-Cantor comb of depth $k$ and $\alpha$
    an optimal bandwidth ordering of $\hat{S}_{b,k}$. Furthermore, let $x$ and
    $y$ be the end vertices of $\hat{S}_{b,k}$ and $B$ the path from $x$ to
    $y$. Then there exists an edge $uv$ of $\hat{S}_{b,k}$ such that
    $I(u,v) \cap I(B)$ is non-empty and $|\alpha(u) - \alpha(v)| =
    \bw(\hat{S}_{b,k})$.
\end{lemma}

\begin{proof}
The graph $\hat{S}_{b,k}$ is a caterpillar with backbone $B$. Let $C_B$ be the
connected component of $\hat{S}_{b,k}[\alpha^{-1}(I(B))]$ that contains $B$.
Observe that $\hat{S}_{b,k} \setminus  C_B$ is a collection of paths, with each
path being a subpath of a stray and having exactly one neighbor in $C_B$.

Let $L$ contain every vertex $u \in N(C_B)$ such that $\alpha(u) < \min[I(B)]$
and $R$ contain every vertex $u \in N(C_B)$ such that $\alpha(u) > \max[I(B)]$.
By definition we have that $L \cup (N(L) \cap C_B)$ induces a matching of size
$L$, such that each matching edge has one endpoint $u$ with $\alpha(u) <
\min[I(B)]$ and the other endpoint $v$ with $\alpha(v) \in \alpha(C_B)$. It follows
that for one of the matching edges $|\alpha(u)-\alpha(v)| \geq |L|$. Thus there
exists an edge $uv$ of $\hat{S}_{b,k}$ such that $I(u,v) \cap I(B)$ is
non-empty and $|\alpha(u) - \alpha(v)| \geq |L|$. An identical argument yields
that there exists an edge $uv$ of $\hat{S}_{b,k}$ such that $I(u,v) \cap I(B)$
is non-empty and $|\alpha(u) - \alpha(v)| \geq |R|$. Thus there exists an edge
$u'y'$ such that $I(u',v') \cap I(B)$ is non-empty and $|\alpha(u') -
\alpha(v')| = \max(|L|,|R|)$.

We now prove that without loss of generality, we can assume that every edge
$uv$ such that neither $u$ nor $v$ are in $C_B$ satisfies $|\alpha(u) -
\alpha(v)| \leq \max(|L|, |R|)$. Let $C_L$ be the set of vertices connected to
$L$ in $G - C_B$ and $C_R$ the set of vertices connected to $R$ in $G - C_B$.
Observe that $C_B, C_L$ and $C_R$ form a partition of $V(\hat{S}_{b,k})$. For
every $v \in C_B \cup L \cup R$ let $\beta(v) = \alpha(v)$. Let $v$ be a vertex
of $C_L \setminus L$ and $u$ the unique vertex of $L$ such that $u$ and $v$ are
connected in $G - C_B$. We then let $\beta(v) = \beta(u) - |L| \cdot
\dist(u,v)$. Handle the vertices of $C_R \setminus R$ symmetrically and let
$\beta'$ be the compressed $\beta$. One can observe that $\beta'$ is a linear
ordering of $\hat{S}_{b,k}$ and that $\bw(\hat{S}_{b,k}, \beta') \leq
\bw(\hat{S}_{b,k}, \alpha) = \bw(\hat{S}_{b,k})$. Clearly, for every edge $uv$
such that neither $u$ nor $v$ are in $C_B$ satisfies $|\alpha(u) - \alpha(v)|
\leq \max(|L|, |R|)$.

Let $uv$ be an edge of $\hat{S}_{b,k}$ such that  $|\alpha(u) - \alpha(v)| =
\bw(\hat{S}_{b,k})$. If one endpoint of $uv$ is mapped to $I(B)$ we are done,
as $uv$ satisfies the conditions of the lemma. On the other hand, if both
endpoints of $uv$ are outside of $I(B)$ then $\bw(\hat{S}_{b,k}) = |\alpha(u) -
\alpha(v)| = \max(|L|,|R|)$. In this case the edge $u'v'$ satisfies the
conditions of the lemma, completing the proof.
\end{proof}

\begin{lemma}
    \label{lemma:skewed-cantor-combs-high-bandwidth}
    For $b \geq k \geq 1$, the bandwidth of any $S_{b,k}$ is at least $k$.
\end{lemma}

\begin{proof}
    The proof of this lemma is inspired by the one for Cantor combs
    given by Chung and Seymour~(\cite{ChungS89}, Lemma 2.1).

    Assume for a contradiction that there is a $\hat{S}_{b,k}$ such that
    $\bw(\hat{S}_{b,k}) < k$. Furthermore, assume without loss of generality
    that $k$ is the smallest such value with respect to $b$. Observe that $k >
    1$. Let $\alpha$ be an ordering of $\hat{S}_{b,k}$ of bandwidth at most
    $k-1$. Let $S, S', P$ and $Q$ be as in the definition of skewed Cantor
    combs. By assumption the bandwidth of both $S$ and $S'$ are $k-1$. Let $x$
    and $y$ be the end vertices of $S$ and $x'$ and $y'$ the end vertices of
    $S'$. Furthermore, let $B$ be the path from $x$ to $y$ and $B'$ the path
    from $x'$ to $y'$. Let $Z$ be the path between the end vertices of
    $\hat{S}_{b,k}$.
    
    Let $\beta$ be the compressed version of $\alpha$ when restricted to $S$.
    Since $\alpha$ is of bandwidth $k-1$, it follows that $\beta$ is of
    bandwidth at most $k-1$ and hence by our assumption $\beta$ is an optimal
    bandwidth ordering of $S$. By
    Lemma~\ref{lemma:skewed-cantor-combs-high-bandwidth-close-to-backbone} we
    know that there exists an edge $uv$ in $S$ such that $I_{\beta}(u,v) \cap
    I_{\beta}(B)$ is non-empty and $|\beta(u)-\beta(v)| = k-1$. It follows
    that $I_{\alpha}(u,v) \cap I_{\alpha}(B)$ is non-empty and
    $|\alpha(u)-\alpha(v)| = k-1$. In the same manner we obtain an edge $u'v'$
    from $S'$. Assume without loss of generality that $\alpha(u) < \alpha(v)$
    and that $\alpha(u') < \alpha(v')$.
    
    Observe that $\alpha^{-1}(I_{\alpha}(u,v)) \subseteq S$ and that
    $\alpha^{-1}(I_{\alpha}(u',v')) \subseteq S'$. It follows directly that the
    inclusion intervals has an empty intersection with $P$. Let $q$ be the
    vertex in $N(Q)$. We can assume without loss of generality that $\alpha(v)
    < \alpha(q)$. There are two cases to consider, either $\alpha(q) <
    \alpha(u')$ or $\alpha(v') < \alpha(q)$.
    
    First we consider the case when $\alpha(q) < \alpha(u')$. Observe that
    $|I(Z)| \leq (k-1)|E(Z)|+1 \leq |V(Q)| + 1$ and $|V(Z)| \geq 5$ since
    $k>1$. It follows from $\alpha(Z) \subseteq I(Z)$ that there is a vertex
    $q' \in Q$ such that $\alpha(q') \notin I(Z)$. Assume without loss of generality
    that $\alpha(q') < \min I(Z)$. It follows that $\alpha(q') < \alpha(u) <
    \alpha(v) < \alpha(q)$. Since there is a path from $q'$ to $q$ disjoint
    from $S$ and $|\alpha(u) - \alpha(v)| = k-1$ it follows that $I(u,v)$ must
    contain a vertex of $Q$, which is a contradiction.

    It remains to consider the case when $\alpha(v') < \alpha(q)$. Observe
    that by assumption $I(u,v)$ and $I(u',v')$ are disjoint. And hence, again we
    consider two cases. First, let $\alpha(v) < \alpha(u')$. We are then in
    the situation that $\alpha(v) < \alpha(u') < \alpha(v') < \alpha(q)$ and
    since there is a path from $v$ to $q$ avoiding $S'$ it follows that this
    path has a non-empty intersection with $I(u',v')$, which is a
    contradiction. The case $\alpha(v') < \alpha(u)$ follows by a symmetric
    argument and hence the proof is complete.
\end{proof}

\subsubsection*{Directions}

\newcommand{\Rstretch}{12b^2}
\newcommand{\Sstretch}{48b^3}
\newcommand{\pos}{\mbox{pos}}
\newcommand{\depth}{\mbox{\em depth}}

Given a caterpillar $T$ and a backbone $B = \left\{ b_1, \dots, b_k \right\}$
we define $pos(P)$ for every stray $P$ in $T$ with respect to $B$, as the
integer $i$ such that $P$ is attached to the vertex $b_i$. Furthermore, we let
$|P|$ denote $|V(P)|$.

\begin{definition}
    Let $T$ be a caterpillar, $B = \left\{ b_1, \dots, b_k \right\}$ a backbone
    of $T$ and $b$ a positive integer. Furthermore, let $\depth$ be a function
    from the strays of $T$ with respect to $B$ to $\mathbb{N}$. For every
    stray $Q$ we let
    \begin{itemize}
        \item $X_Q = \left\{ P \mid \pos(P) +
            \frac{|P|}{2b} < \pos(Q) \mbox{ and } \pos(Q) -
            \frac{|Q|}{2b} \leq \pos(P) - \frac{|P|}{2b}
        \right\}$ and
    \item $Y_Q = \left\{ P \mid \pos(Q) < \pos(P) -
            \frac{|P|}{2b} \mbox{ and } \pos(P) +
            \frac{|P|}{2b} \leq \pos(Q) + \frac{|Q|}{2b}
        \right\}$.
    \end{itemize}
   Let $x_Q = \max(\depth(X_Q))$ and $y_Q = \max(\depth(Y_Q))$. We say that $Q$ is
   \emph{pushed east} if $x_Q > y_Q$, \emph{pushed west} if $x_Q < y_Q$ and
   \emph{lifted} if $x_Q = y_Q$.
\end{definition}

We say that a skewed $b$-Cantor comb of depth $k$ is centered around the stray
$Q$, where $Q$ is as in the definition of $S_{b,k}$.  For a caterpillar $T$ we
say that a \emph{backbone $B$ is maximized} if for every other backbone $B'$
it holds that $|B'| \leq |B|$.

We will now describe an algorithm $\cantoralg$ that given a caterpillar $T$, a
maximized backbone $B$ of $T$ and a positive integer $b$ searches for skewed
Cantor combs in $T$. Let $\depth$ be a function from the strays of $T$ with
respect to $B$ into $\mathbb{N}$. As an invariant, $\depth$ promises there to
be a skewed $(b+1)$-Cantor comb centered around $Q$ of depth $\depth(Q)$. The
exception is if $\depth(Q)$ is $0$, then the stray is so short that we ignore
it and we hence make no promises with respect to skewed $(b+1)$-Cantor combs.
Initially, for every stray $Q$ let $\depth(Q)$ be $2$ if $|Q| \geq 4b$ and $0$
otherwise. Observe that the invariant is true due to $B$ being a maximized
backbone.

Now we search for a stray $Q$ that is lifted such that both $x_Q$ and $y_Q$
are at least $\depth(Q)$. It such a $Q$ is found, increase $\depth(Q)$ by one.
Observe that there is in fact a skewed $(b+1)$-Cantor comb centered around $Q$
of this depth ($\depth(Q)$ after the incrementing). Run this procedure until
such a stray $Q$ can not be found or until $\depth(Q)$ reaches $b+1$ for some
stray. Observe that we can for every stray evaluate $x_Q$ and $y_Q$ in
$O(n^2)$. And since this is done at most $O(bn)$ times, the running time of
$\cantoralg$ is bounded by $O(bn^3)$.

The reader should note that $\cantoralg$ does not detect all skewed $b$-Cantor
combs. In fact, it searches only for a stricter version and might overlook the
deep skewed $(b+1)$-Cantor combs in a caterpillar. But, as it turns out, these
stricter versions are sufficient for our purposes. From now on, we will assume
that the function applied when evaluation whether a stray is pushed west or
east, is the depth function calculated by running $\cantoralg$.

\begin{definition}
    For a caterpillar $T$, a maximized backbone $B = \left\{ b_1, \dots, b_l
    \right\}$ of $T$ and a positive integer $b$ we define the \emph{directional
    stray graph} as the following interval graph: for every stray $P$ add the
    interval 

    \begin{itemize}
        \item $[\pos(P)\Sstretch - \Rstretch|P|, \pos(P)\Sstretch]$ if $P$ is
            pushed west and
        \item $[\pos(P)\Sstretch, \pos(P)\Sstretch + \Rstretch |P|]$ otherwise. 
    \end{itemize}
    We say that an interval originating from a stray pushed west is
    \emph{west oriented} and visa versa.
\end{definition}

\begin{lemma}
    \label{lemma:few-long-but-close-interval}
    Let $T$ be a caterpillar, $b$ a positive integer, $G_I$ some directional stray
    graph of $T$ and $x$ and $y$ two natural numbers such that $x < y$. Then
    either there are at most $2b$ intervals of length at least $y-x$ in $G_I$
    starting within $[x,y]$, or $\bw(T) > b$.
\end{lemma}

\begin{proof}

    Assume otherwise for a contradiction and let $\bw(T) \leq b$ and $K$ be a
    set of $2b+1$ intervals of length at least $y-x$ starting within $[x,y]$.
    Let $x'$ be the smallest number such that $x \leq x'$ and $x'$ is
    divisible by $\Sstretch$ and $y'$ the largest number such that $y' \leq y$
    and $y'$ is divisible by $\Sstretch$. Observe that all intervals in $K$
    has their starting point within $[x', y']$ by construction. Consider the
    minimum connected, induced subgraph $H$ of $T$ containing the vertices of
    the strays corresponding to the intervals in $K$. We will consider $H$
    with respect to the backbone such that the strays of $H$ are exactly the
    ones corresponding to intervals in $K$. Let $z = y'-x'$ and observe that
    every stray in $H$ contains at least $q = z/\Rstretch$ vertices and that
    the backbone of $H$ is of length $r = z/\Sstretch$.  It follows that

    \begin{align*}
        D(G) &\geq \frac{|V(H)|-1}{\diam(H)} \\
        &\geq \frac{(2b+1)q + r -1}{2q + r} \\
        &> \frac{2bq+r}{2q+r} \\
        &\geq b \frac{2q+r/b}{2q/b+r/b} \\
        &\geq b
    \end{align*}
    
    \noindent
    which contradicts $D(G) \leq b$ and hence we know that there are at most
    $2b$ such intervals. Note that we used the fact that $q > 1$. This follows
    from the fact that $x' < y'$ due to the local density bound and hence $q
    \geq \Sstretch/\Rstretch \geq 4$.

\end{proof}

\begin{lemma}
    \label{lemma:chromatic-number-directional-stray-graph}
    Let $T$ be a caterpillar, $b$ a positive integer and $G_I$ some directional stray
    graph of $T$. Then either $\chi(G_I) < \Rstretch$ or $\bw(T) > b$.
\end{lemma}

\begin{proof}

    Assume for a contradiction that $\chi(G_I) \geq 12b^2$ and that $\bw(T)
    \leq b$. Then there is a number $w$ such that at least $12b^2$ of the
    intervals of $G_I$ contains $w$. This follows from the well-known result
    that $\chi(G_I)$ equals the size of the maximum clique of $G_I$, since
    $G_I$ is an interval graph. Let $I$ be the set of all east oriented
    intervals containing $w$ and assume without loss of generality that $I$ is
    of size at least $6b^2$. Discard the elements of $I$ with the highest
    starting value and let $[x', y']$ be a discarded element. Observe that at
    most $2b$ elements were discarded due to the local density bound.  Hence
    we now have at least $6b^2-2b$ elements left.  We will start by giving a
    lower bound on the length of the intervals in $I$.  Consider an element
    $[x, y]$ of shortest length in $I$. By definition $x < x' \leq y$ and by
    construction $x'-x \geq \Sstretch$, hence $y-x \geq \Sstretch$ and it follows
    that all elements of $I$ are of length at least $\Sstretch$.
    
    Let $[x_2, y_2]$ be a shortest interval in $I$ and recall that the stray
    $P^2$ corresponding to the interval is attached to the backbone vertex
    $b_{c_2}$ for $c_2 = x_2/\Sstretch$. Furthermore, $|P^2| =
    (y_2 - x_2)/\Rstretch \geq \Sstretch / \Rstretch = 4b$. Since the
    backbone used when constructing $G_I$ is maximized it follows that the
    distance from $b_{c_2}$ to any endpoint of the backbone is at least $4$
    and hence there is an $S_{b+1,2}$ centered around $b_{c_2}$. 

    Discard all intervals with their starting point within $[x_2 - 2(y_2-x_2),
    y_2]$ in $I$. We know that at most $6b$ elements are discarded by
    Lemma~\ref{lemma:few-long-but-close-interval}.  Now let
    $[x_3, y_3]$ be a shortest interval in $I$ and recall that the stray
    $P^3$ corresponding to the interval is attached to the backbone vertex
    $b_{c_3}$ for $c_3 = x_3/(\Rstretch)$.  Observe that $|y_3 - x_3| > |x_3 -
    x_2|$ and that $|y_2 - x_2| < \frac{1}{2}|x_3-x_2|$ and hence
    
    \begin{gather*}
    |y_3-x_3| > |x_3-x_2| > \frac{1}{2}|x_3-x_2| + |y_2-x_2| \\
    \implies \\
    \frac{|y_3-x_3|}{2b(\Rstretch)} > \frac{|x_3-x_2|}{\Sstretch} +
    \frac{|y_2-x_2|}{2b(\Rstretch)} \\
    \implies \\
    \frac{|V(P^3)|}{2b} > |c_3 - c_2| + \frac{|V(P^2)|}{2b}.
    \end{gather*}

    Let $S$ be the $S_{b+1,2}$ centered around $b_{c_2}$ and recall that by
    definition the distance from $b_{c_2}$ to any backbone vertex of $S$ is
    bounded from above by $\frac{|V(P^2)|}{2b}$.  It follows that the distance
    from $b_{c_3}$ to any backbone vertex of $S$ is bounded by
    $\frac{|V(P^3)|}{2b}$. Since $[x_3, y_3]$ is east oriented  there is
    another $S_{b+1,i}$ centered around a stray $\bar{P}^2$ such that
    $\pos(\bar{P}^2) + \frac{|\bar{P}^2|}{2b} < c_3$ and $\pos(\bar{P}^2) -
    \frac{|\bar{P}^2|}{2b} \geq c_3 - \frac{|P^3|}{2b}$ for some $i \geq 2$.
    By definition, the $S_{b+1,i}$ contains an $S_{b+1,2}$ as a subgraph in
    such a way that there is an $S_{b+1,3}$ centered around $c_3$.  Discard all
    intervals with starting points within $[x_3 - 2(y_3-x_3), x_3]$ and repeat
    the argument to obtain a $S_{b+1,4}$. We keep repeating the argument until
    we obtain a $S_{b+1,b+1}$

    Notice that we can do this as we are discarding at most $6b$ vertices each
    time, repeating the procedure $b-1$ times and $I$ contains at least
    $6b^2-2b > 6b(b-1)$ intervals. This completes the proof, as we know from
    Lemma~\ref{lemma:skewed-cantor-combs-high-bandwidth} that $\bw(S_{b+1,b+1})
    \geq b+1$.
\end{proof}

\subsubsection*{Algorithm and Correctness}

\begin{algorithm}[ht!]
  \KwIn{A caterpillar $T$ and a positive integer $b$.}
  \KwOut{A $\catar$-bandwidth ordering of $T$ or conclusion that $\bw(T) > b$.}
\BlankLine

Let $B = \left\{ b_1, \dots, b_k \right\}$ be a maximized backbone of $T$.\\
Construct the directional stray graph $G_I$ of $T$ with respect to $B$.\\
Find a minimum coloring of $G_I$.\\
\If{$\chi(G_I) \geq \Rstretch$} {
    \KwRet $\false$.
}
Let $\alpha(b_i) = \Sstretch(n+i)$. \\
Let $\mathcal{P}$ be the collection of strays in $T$ with respect to $B$.\\
For every stray $P$ in $\mathcal{P}$ let $C(P)$ be the color of the interval representing the stray.\\
\For{every $P \in \mathcal{P}$}{
    Let $p_1, \dots, p_k$ be the vertices of $P$ such that $\dist(B, p_i) <
    \dist(B, p_{i+1})$ for every $i$.\\
    Let $\left\{ u \right\} = N(P)$.\\
    \If{$P$ is pushed west} {
        Let $\alpha(p_i) = \alpha(u) + C(P) - i\Rstretch$ for
        every $i$.
    }
    \Else {
        Let $\alpha(p_i) = \alpha(u) + C(P) + (i-1)\Rstretch$ for
        every $i$.
    }
}
\KwRet Compressed version of $\alpha$.

  \caption{$\catalg$}
  \label{alg:catalg}
\end{algorithm}

\begin{theorem}
    \label{theorem:caterpillar-bandwidth-fptapprox}
    There exists an algorithm that given a caterpillar $T$ and a positive integer $b$
    either returns a $\catar$-bandwidth ordering of $T$ or correctly
    concludes that $\bw(T) > b$ in time $O(bn^3)$.
\end{theorem}

\begin{proof}
    
    Recall that $\cantoralg$ runs in $O(bn^3)$ time. Furthermore, a coloring of
    $G_I$ can be found in $O(n)$ time by
    Golumbic~\cite{golumbic2004algorithmic}. Observe that every other step of
    the algorithm trivially runs in $O(n)$ time. And hence the algorithm runs
    in $O(bn^3)$ time. If $\catalg$ returns $\false$, then $\chi(G_I) \geq
    \Rstretch$. It follows from
    Lemma~\ref{lemma:chromatic-number-directional-stray-graph} that $\bw(T) > b$
    and hence the conclusion is correct. We will now prove that $\alpha$ is a
    sparse ordering of $V(T)$ of bandwidth at most $\Sstretch$.  It is clear
    that for any edge $uv \in E(T)$ it holds that $|\alpha(u) - \alpha(v)| \leq
    \Sstretch$. It remains to prove that $\alpha$ is an injective function.
    Assume for a contradiction that there are two vertices $u, v$ such that
    $\alpha(u) = \alpha(v)$. Observe that $\alpha(u) \equiv 0
    \bmod{(\Sstretch)}$ if and only if $u$ is a backbone vertex of $T$. This
    comes from the fact that $\chi(G_I) < 12b^2$. And since it is clear from
    the algorithm that no two vertices of the backbone are given the same
    position we can assume that neither $u$ nor $v$ is a backbone vertex. It
    follows that $\alpha(u) \equiv c(P) \bmod{(\Rstretch)}$ where $P$ is the
    stray containing $u$. Observe that the algorithm gives unique positions to
    all vertices from the same stray and hence $u$ and $v$ must belong to two
    different strays given the same color. Let $P_u$ be the stray containing
    $u$ and $P_ v$ the strain containing $v$. Furthermore, let $[x_u, y_u]$ and
    $[x_v, y_v]$ be the corresponding intervals in $G_I$. Observe that $I(P_u)
    \subseteq [x_u, y_u]$ and $I(P_v) \subseteq [x_v, y_v]$ and hence $[x_u,
    y_u] \cap [x_y, y_v] \neq \emptyset$, which is a contradiction, completing the proof.
\end{proof}
\section{Concluding Remarks}
\label{section:conclusion}

We have shown that the classical $2^{O(b)}n^{b+1}$ time algorithm of
Saxe~\cite{saxe1980dynamic} for the {\sc Bandwidth} problem is essentially
optimal, even on trees of pathwidth at most $2$. On trees of pathwidth $1$,
namely caterpillars with hair length $1$, the problem is known to be polynomial
time solvable. On the positive side, we gave the first approximation algorithm
for {\sc Bandwidth} on trees  with approximation ratio being a function of $b$
and independent of $n$. Our approximation algorithm is based on pathwidth,
local density and a new obstruction to bounded bandwidth called skewed Cantor
combs. We conclude with a few open problems.

\begin{itemize}
    \item Does {\sc Bandwidth} admit a parameterized approximation algorithm on
        general graphs?
    \item Does {\sc Bandwidth} admit an approximation algorithm on trees with
        approximation ratio polynomial in $b$? What if one allows the algorithm
        to have running time $f(b)n^{O(1)}$?
    \item Does there exist a function $f$ such that any graph $G$ with
        pathwidth at most $c_1$, local density at most $c_2$, and containing no
        $S_{c_3,c_3}$ as a subgraph has bandwidth at most $f(c_1, c_2, c_3)$?
\end{itemize}

\end{document}